\documentclass[letterpaper, 10pt, conference]{ieeeconf}

\IEEEoverridecommandlockouts
\usepackage[utf8]{inputenc}
\usepackage{tikz}

\usepackage{amsmath}
\usepackage{amsfonts}
\usepackage{epsfig}
\usepackage{caption}
\usepackage{subcaption}
\usepackage{MnSymbol}
\usepackage{graphicx}
\usepackage{xcolor}
\usepackage{tikz}
\usepackage{cite}
\usepackage{mathtools}
\usepackage{epstopdf}
\DeclareGraphicsExtensions{.pdf}
\newtheorem{definition}{Definition}

\newtheorem{theorem}{Theorem}
\title{\LARGE \bf
{Almost-global tracking for a rigid body with internal rotors
}}
\author{A. Nayak$^{1}$, R. N. Banavar$^{2}$ and D. H. S. Maithripala $^{3}$
\thanks{$^{1}$ A. Nayak is with Systems and Control Engineering,
Indian Institute of Technology Bombay, India. {\tt\small aradhana@sc.iitb.ac.in}}
\thanks{$^{2}$R. N. Banavar is currently visiting professor at Department of Electrical Engineering, Indian Institute of Technology Gandhinagar, India.{\tt\small banavar@iitb.ac.in}}
\thanks {$^{3}$ D. H. S. Maithripala
	is with the University of Peradeniya, Sri Lanka. {\tt \small smaithri@pdn.ac.lk}
        }%
}

\begin{document}

\maketitle
\thispagestyle{empty}
\pagestyle{empty}


\begin{abstract}
%
Almost-global orientation trajectory tracking for a rigid body with external actuation has been well studied in the literature, and in the geometric setting as well. The tracking control law relies on the fact that a rigid body is a simple mechanical system (SMS) on the $3-$dimensional group
of special orthogonal matrices. However, the
problem of designing feedback control laws for tracking using internal actuation mechanisms, like rotors or control moment gyros, has received
lesser attention from a geometric point of view. An internally actuated rigid body is not a simple mechanical system, and the phase-space here evolves on the
level set of a momentum map. In this note, we propose a  novel
proportional integral derivative (PID) control law for a rigid body with $3$ internal rotors, that achieves tracking of feasible trajectories from almost all initial conditions.
\end{abstract}

\section{Introduction}
Spacecrafts are actuated either through internal or external mechanisms. External mechanisms
include gas jet thrusters while internal mechanisms include spinning rotors and control moment
gyros. In the recent past, there has been an increased interest in the design of coordinate-free control laws for simple mechanical systems (\cite{bulo}, defined in section II) which evolve on Lie groups. Results on stabilization of a rigid body, which is an SMS on $SO(3)$, about a desired configuration in $SO(3)$ using proportional plus derivative (PD) control are found in \cite{ejc}, \cite{crouch}, \cite{buloandmurray}. Geometric tracking of specific mechanical systems such as a quadrotor, which is an SMS on $SE(3)$, and a rigid body, can be found in \cite{tleemleok}, \cite{tlee}. Almost-global tracking of a reference trajectory for an SMS on a Lie group implies tracking of the reference from almost all initial conditions in the tangent bundle of the Lie group.  A general result on almost-global asymptotic tracking (AGAT) for an SMS on a class of compact Lie groups is found in \cite{dayawansa}. In all these results, the rigid body is assumed to be externally actuated and the control torque is supplied through actuators such as gas jets. Almost-global stabilization and tracking of the externally actuated rigid body is, therefore, a sufficiently well studied problem. However, the problem of geometric tracking for an internally actuated rigid body has received much less attention.

Interconnected mechanical systems have been studied in the context of spherical mobile robots in \cite{sneha}, \cite{cai}, \cite{karanaev} et al.The stabilization of the internally actuated rigid body is studied in \cite{bkmm}. It is shown that any feedback torque on the externally actuated rigid body can be realised with $3$ internal rotors attached to the rigid body. This paper
illustrates the fact that despite the feedback forces acting on an externally actuated rigid body, it is still Hamiltonian and behaves like a heavy rigid body. In other words, if a certain class of feedback torque is applied to the rotors, the rigid body with the rotors is a fully actuated SMS on $SO(3)$. This motivates us to study the AGAT problem for a rigid body with rotors as AGAT for an SMS has been studied extensively (\cite{dayawansa},\cite{anrnb2}, \cite{pidmtp}).
%
%
\newline
In \cite{ejc} and \cite{weiss} the almost-global asymptotic stabilization (AGAS) problem of a rigid body with $3$ internally mounted rotors is solved using proportional plus derivative (PD) control. In \cite{hall}, rigid body tracking is achieved using both external and internal actuation using local representation for rotation matrices. In \cite{madhumtp}, the trajectory tracking problem is considered for a hoop robot with internal actuation such as a pendulum. It is shown that a class of internal actuation configurations exist for which the underactuated mechanical system can be converted to a fully actuated SMS by feedback torques.
\newline
AGAT of an SMS on a Lie group is often achieved by a proportional plus derivative type control(\cite{anrnb2}, \cite{dayawansa}). A configuration error is chosen on the Lie group with the help of the group operation along with a compatible navigation function. A navigation function is a Morse function with a unique minimum. The closed loop error dynamics,  for a control force proportional to the negative gradient covector field generated by the navigation function plus a dissipative covector field, is then an SMS,. This control drives the error dynamics to the lifted minimum of the navigation function on the tangent bundle of the Lie group from all but the lifted saddle points and maxima of the navigation function. As the critical points of a Morse function are isolated, this convergence is almost global. The compatibility conditions in \cite{anrnb2} ensure that the error function is symmetric and achieves its minimum when two configurations coincide. In \cite{pidmtp}, the authors propose an 'integral' action to the existing PD control law in \cite{dayawansa}. The addition of an integral term makes the control law robust to bounded parametric uncertainty.
%
%
%
%
%
%

A rigid body with internal rotors is an underactuated, interconnected simple mechanical system. The control torque provided to the rotors gets reflected through the interconnection mechanism to the rigid body. Due to absence of external forces, the total angular momentum is conserved. This restriction implies that only a certain class of angular velocities can be attained at any configuration in $SO(3)$. Also, due to the presence of quadratic rotor velocity terms, the rigid body alone is not an SMS. We isolate the rigid body dynamics by introduction of feedback control terms in the system dynamics so that the closed loop rigid body dynamics is a fully actuated SMS. Thereafter we apply the existing AGAT control to the rigid body and obtain the corresponding rotor trajectories from the rotor dynamics. As the control objective is to track a suitable reference trajectory on $SO(3)$, the rotor speeds are allowed to be arbitrary. The paper is organised as follows- in section II, after
presenting a few mathematical preliminaries, we derive equations for the rigid body with external actuation and with $3$ internal rotors. In section III we append an integral term to the control law for AGAT of an SMS on a Lie group in \cite{anrnb2} and propose the AGAT control for the rigid body with $3$ rotors for an admissible class of reference trajectories. In section IV we present simulation results for the proposed control law.


\section{Preliminaries}
This section introduces conventional mathematical notions to describe simple mechanical systems which can be found in \cite{bulo}, \cite{mars}, \cite{arnold}. A Riemannian manifold is denoted by the 2-tuple $(Q, \mathbb{G})$, where $Q$ is a smooth connected manifold and $\mathbb{G}$ is the metric on $Q$. $\nabla$ denotes the Riemannian connection on $(Q, \mathbb{G})$ (\cite{petersen},\cite{yano}).The flat map $\mathbb{G}^{\flat} :  T_q Q \to T_q ^* Q$ is given by $ \mathbb{G}(v_1, v_2) = \langle \mathbb{G}^{\flat} (v_1); v_2 \rangle$ for $v_1$,$v_2 \in T_q Q$ and the sharp map is its dual $ \mathbb{G}^{\sharp}: T_q ^* Q \to T_q Q$, and given by $ \mathbb{G}^{-1}(w_1, w_2)= \langle \mathbb{G}^{\sharp}(w_1); w_2 \rangle$ where $w_1$, $w_2 \in T_q ^ * Q$. Therefore if $\{ e^i\}$ is a basis for $T^*_q Q$, $ \mathbb{G}^{\flat} (v_1) = \mathbb{G}_{ij} {v_1}^j e^i$ and $ \mathbb{G}^{\sharp} (w_1) = \mathbb{G}^{ij} {w_1}_j e_i.$
\begin{definition}(SMS)
  A simple mechanical system (or SMS) on a smooth manifold $Q$ with a metric $\mathbb{G}$ is denoted by the 7-tuple $(Q,\mathbb{G},V, F, \mathcal{F}, U)$, where $V$ is a potential function on $Q$, $F$ is an external uncontrolled force, $\mathcal{F}= \{F^1 \dotsc F^m\}$ is a collection of covector fields on $Q$, and $U \subset \mathbb{R}^m$ is the control set. The system is fully actuated if $T_{q}^{*} Q = span\{\mathcal{F}_{q}\}$, $\forall q \in Q$. The governing equations for the above SMS without any control input is given by
\begin{equation}\label{dynrm}
\stackrel{\mathbb{G}}{\nabla}_{\dot{\gamma}(t)} \dot{\gamma}(t) = -grad V(\gamma(t)) +\mathbb{G}^{\sharp} (F(\dot{\gamma}(t)))
\end{equation}
where $grad V(\gamma(t)) = \mathbb{G}^{\sharp}\mathrm{d}V(\gamma(t))$ and $\gamma(t)$ is the system trajectory.
\end{definition}
%
%
%
%
Let $G$ be a Lie group and let $\mathfrak{g}$ denote its Lie algebra. Let $\phi: G \times G \to G$ be the left group action in the first argument defined as $\phi(g,h) \coloneq  L_{g} (h)= gh $ for all $g$, $h \in G$. The infinitesimal generator corresponding to $\xi \in \mathfrak{g}$ is $\xi_Q \in \Gamma ^{\infty} (TQ)$ which is defined as $\xi_Q(q) = \frac{\mathrm{d}}{\mathrm{d}t}|_{t=0} \phi (exp(t\xi), q)$, where $exp$ denotes the exponential map. The Lie bracket on $\mathfrak{g}$ is $[,]$. The \textit{adjoint map}, $ad_\xi : \mathfrak{g} \to \mathfrak{g}$ for $\xi \in \mathfrak{g}$ is defined as $ad_\xi \eta \coloneq [\xi, \eta]$ for $\eta \in \mathfrak{g}$. Let $\mathbb{I} :\mathfrak{g} \to \mathfrak{g}^*$ be an isomorphism from the Lie algebra to its dual. The inverse is denoted by $\mathbb{I}^\sharp: \mathfrak{g}^* \to \mathfrak{g}$. $\mathbb{I}$ induces a
left invariant metric on $G$ (\cite{bulo}), which we denote by $\mathbb{G}_{\mathbb{I}}$ and define by the following $\mathbb{G}_{\mathbb{I}}(g).(X_g,Y_g) = \langle \mathbb{I}(T_gL_{g^{-1}} (X_g)),T_gL_{g^{-1}} (Y_g)\rangle$ for all $g \in G$ and $X_g$, $Y_g \in T_g G$. The equations of motion for the SMS $(G, {\mathbb{I}},F)$ where $F \in \mathfrak{g}^*$ are derived from \eqref{dynrm} given by
\begin{align}\label{dynliegrp}
\xi &= T_g L_{g^{-1}} \dot{g},\\ \nonumber
\dot{\xi} - \mathbb{I}^\sharp ad^*_\xi \mathbb{I} \xi &= \mathbb{I}^{\sharp} F
\end{align}
where $g(t)$ describes the system trajectory. $\xi(t)$ is called the \textit{body velocity} of $g(t)$.
\subsection{Dynamics of a rigid body with external actuation}
Consider a rigid body with external actuation provided through gas jets mounted on the principal axes. Let $I \in \mathbb{R}^{3 \times 3}$ denote the moment of inertia of the rigid body in the body frame, $u \in \mathfrak{so(3)} \sim \mathbb{R}^3 $ be the control vector field applied to the gas jets, $R(t) \in SO(3)$ be the system trajectory and $\Omega(t)$ be the body velocity of $R(t)$. As this is a SMS given by $(SO(3), I, Iu)$, the equations of motion are given by \eqref{dynliegrp} with $F = Iu$ as follows
\begin{subequations}\label{dynrigext}
\begin{equation}
\dot{R}= R \widehat{\Omega}
\end{equation}
\begin{equation}
\dot{\Omega}-I^\sharp ad^*_\Omega I \Omega = u_{ext}
\end{equation}
\end{subequations}

\subsection{Dynamics of a rigid body with internal actuation}
\begin{figure}[h!]
\centering
\includegraphics[scale=0.4]{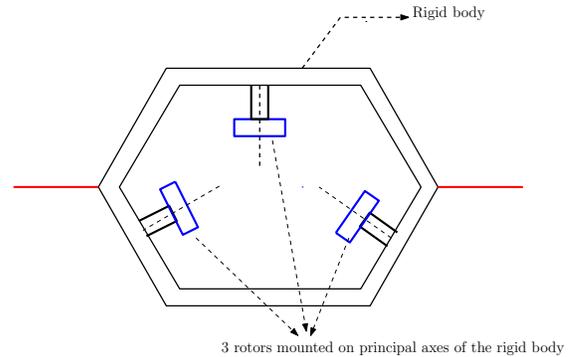}
\caption{Rigid body with 3 rotors}\label{figrigrot}
\end{figure}
In the internal actuation case, for a rigid body with 3 rotors,
 the configuration space is $Q=SO(3) \times S^1 \times S^1 \times S^1$ and the configuration variable is denoted as $q= (R, \Theta)$ where $R \in SO(3)$ and $\Theta= \begin{pmatrix}
           \theta_1 & \theta_2 & \theta_3
         \end{pmatrix}^T$, $\theta_i \in S^1$ for $i= 1, 2,3$. The rotors are assumed to be mounted on the principal axes of the rigid body as shown in Figure ~\ref{figrigrot}. The moment of inertia of the rigid body is $I\in \mathbb{R}^{3 \times 3}$  in the rigid body frame. $K= diag(k_1, k_2,k_3)$ is the inertia matrix of the $3$ rotors in the rigid body frame, where $k_i$ is moment of inertia of $i$th rotor, $i =1,2,3$. In the absence of potential energy, the Lagrangian is chosen to be the kinetic energy and given as
\begin{equation}\label{L}
  L(q, \dot{q}) = \frac{1}{2}\langle\Omega , I \Omega \rangle + \frac{1}{2}\langle \Omega+ \Omega_r, K(\Omega+ \Omega_r) \rangle
\end{equation}
 where $\hat{\Omega} = {R^{-1} \dot{R}}$ and $\Omega_r = \begin{pmatrix}
                                                           \dot{\theta}_1 & \dot{\theta}_2 & \dot{\theta}_3
                                                         \end{pmatrix}^T$.
The manifold $Q$ is a trivial principal $G-$ bundle as $Q= G \times S$ where the Lie group is $G= SO(3)$ and the shape space is $S = S^1 \times S^1 \times S^1$. Using the trivialization in $Q$ it can be shown that $(TQ)$ is locally diffeomorphic to $TS \times \mathfrak{g} \times G$. Therefore, $(TQ)/G$ is diffeomorphic to $TS \times \mathfrak{g}$. Further, as the lagrangian in \eqref{L} is invariant under the action of the $G$, it reduces from a function on $TQ$ to a function on $(TQ )/ G $. By the local trivialization coordinates on $(TQ )/ G $ are $ (\Theta, \Omega_r, \Omega)$ where, $(\Theta, \Omega_r)$ are coordinates for $TS$ and $\Omega$ is the coordinate for $\mathfrak{g}$. Therefore, the reduced Lagrangian $l : (TQ)/G \to \mathbb{R}$ is
\begin{equation}\label{l}
  l(\Theta, \dot{\Theta}, \Omega)= \frac{1}{2}\langle\Omega , I \Omega \rangle + \frac{1}{2}\langle \Omega+ \Omega_r, K(\Omega+ \Omega_r) \rangle
\end{equation}
The variational principle with the reduced Lagrangian is applied by dividing variations $\delta q$ of $q$ into those only in $\Theta$ and those only in $R$. In $\Theta$, we get the usual Euler Lagrange equations, while in $R$, we obtain Euler Poincare equations. The equations of motion are together called \textit{Hamel equations} (\cite{marskrish})
\begin{subequations}\label{dynrigrot}
  \begin{equation}
    \frac{\mathrm{d}}{\mathrm{d}t} \bigg( \frac{\partial l}{\partial \Omega}\bigg) = ad^*_\Omega \bigg( \frac{\partial l}{\partial \Omega} \bigg),
  \end{equation}
  \begin{equation}
    \frac{\mathrm{d}}{\mathrm{d}t} \bigg(  \frac{\partial l}{\partial \Omega_r} \bigg) +  \frac{\partial l}{\partial \Theta} =u_{int}
  \end{equation}
\end{subequations}
where $u_{int} \in T_\Theta( S^1 \times S^1 \times  S^1)$ is the control input  applied to the rotors and the reconstruction equation is given as $\hat{\Omega }= {R^{-1} \dot{R}}$. Substituting for $l$ from \eqref{l} in \eqref{dynrigrot},
\begin{subequations}\label{redynrigrot}
  \begin{equation}
    (I+K) \dot{\Omega} + K\dot{\Omega}_r = \Pi \times \Omega,
  \end{equation}
  \begin{equation}
   K(\dot{\Omega} + \dot{\Omega}_r) =u_{int}
  \end{equation}
\end{subequations}
where $\Pi =(I+K) \Omega+K \Omega_r $ is the body angular momentum. The reconstruction equation for attitude of spacecraft and angular displacement of rotors is given by
\[\dot{R} = R \hat{\Omega}, \qquad \dot{\Theta}= \Omega_r\]
respectively.\\
\textit{Remark:} \eqref{redynrigrot} can be written as
\begin{equation}\label{form}
  \begin{pmatrix}
    I+K & K \\
    K & K
  \end{pmatrix} \begin{pmatrix}
                  \dot{\Omega} \\
                  \dot{\Omega_r}
                \end{pmatrix} = \begin{pmatrix}
                  \Pi \times {\Omega} \\
                  0
                \end{pmatrix}+ \begin{pmatrix}
                                 0 \\
                                 u_{int}
                               \end{pmatrix}
\end{equation}
Therefore, \eqref{form} is an underactuated SMS.

\section{Geometric objects and admissible trajectories}
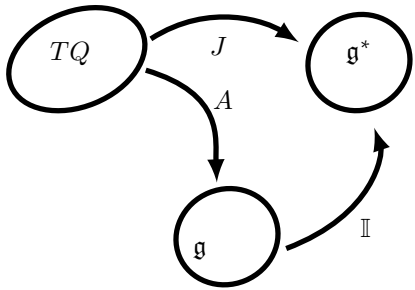
\begin{figure}[ht]
\centering
\begin{tikzpicture}[y=0.80pt, x=0.80pt, yscale=-0.4, xscale=0.4, inner sep=0pt, outer sep=0pt]
\begin{scope}[shift={(0,-698.0315)}]
  \path[cm={{0.92509,-0.37974,0.37974,0.92509,(0.0,0.0)}},draw=black,miter
    limit=4.00,line width=2.000pt] (-202.1721,751.9123) ellipse (2.3566cm and
    1.5804cm);
  \path[cm={{0.92509,-0.37974,0.37974,0.92509,(0.0,0.0)}},draw=black,miter
    limit=4.00,line width=2.000pt] (104.8280,882.2546) ellipse (1.7132cm and
    1.5804cm);
  \path[cm={{0.92509,-0.37974,0.37974,0.92509,(0.0,0.0)}},draw=black,miter
    limit=4.00,line width=2.000pt] (-117.7124,1013.5853) ellipse (1.7132cm and
    1.5804cm);
  \path[-latex, draw=black,line join=miter,line cap=butt,miter limit=4.00,even odd
    rule,line width=2.000pt] (186.0000,748.3622) .. controls (256.0000,703.3622)
    and (324.0000,728.3622) .. (361.0000,752.3622);
  \path[-latex, draw=black,line join=miter,line cap=butt,miter limit=4.00,even odd
    rule,line width=2.000pt] (346.0000,996.3622) .. controls (452.0000,956.3622)
    and (460.0000,883.3622) .. (450.0000,853.3622);
  \path[-latex, draw=black,line join=miter,line cap=butt,miter limit=4.00,even odd
    rule,line width=2.000pt] (180.0000,785.3622) .. controls (274.0000,813.3622)
    and (262.0000,867.3622) .. (264.0000,918.3622);
  \path[fill=black,line join=miter,line cap=butt,line width=0.800pt]
    (67.0000,752.3622) node[below right] (text4195) {$TQ$};
  \path[fill=black,line join=miter,line cap=butt,line width=0.800pt]
    (416.9646,777.9581) node[above right] (text4195-7) {$\mathfrak{g}^{*}$};
  \path[fill=black,line join=miter,line cap=butt,line width=0.800pt]
    (235.9499,987.5004) node[below right] (text4195-7-3) {$\mathfrak{g}$};
  \path[fill=black,line join=miter,line cap=butt,line width=0.800pt]
    (259.7468,831.2916) node[above right] (text4195-6-9) {$A$};
  \path[fill=black,line join=miter,line cap=butt,line width=0.800pt]
    (255.7468,746.2916) node[below right] (text4195-6-2) {$J$};
  \path[fill=black,line join=miter,line cap=butt,line width=0.800pt]
    (433.9362,961.2916) node[below right] (text4195-6-2-3) {$\mathbb{I}$};
\end{scope}

\end{tikzpicture}
\caption{Relationship between the mechanical connection, locked inertia tensor and momentum map in a principle fibre bundle} \label{fig19}
\end{figure}
The momentum map $J : T(SO(3) \times S^1 \times S^1\times S^1) \to \mathfrak{so(3)}^*$ gives the conserved quantity along trajectories to \eqref{dynrigrot} as the Lagrangian (in \eqref{L}) is invariant with respect to action of $SO(3)$. $J: TQ \to \mathfrak{g}^*$ is defined as
\begin{equation}\label{Jdef}
  \langle J(q, v), \xi \rangle= \ll v, \xi_Q(q)\gg
\end{equation}
for $(q,v) \in TQ$ and $\xi \in \mathfrak{g}$.
The \textit{mechanical connection} $A: TQ \to \mathfrak{g}$ is expressed in terms of connection coefficient $A(\Theta)$ in the body frame as
\begin{equation}\label{condef}
  A(\Theta, R, \Omega_r, \Omega)= A(\Theta) \Omega_r + \Omega
\end{equation}
$A(\Theta): T_{\Theta} (S^1 \times S^1 \times S^1) \to \mathfrak{so(3)}$ is the obtained from \eqref{L} as $A(\Theta) = (I+K)^{-1}K$ (details in \cite{ostrowski}) and therefore,
\begin{equation}\label{con}
  A(\Theta, R,\Omega_r, \Omega)= (I+K)^{-1}K \Omega_r + \Omega
\end{equation}
The \textit{locked inertia tensor} is $\mathbb{I}(\Theta, R): \mathfrak{so(3)} \to \mathfrak{so(3)}^*$ in body frame is the obtained from \eqref{L} as $\mathbb{I}= (I+K)$ and in the inertial frame as $\mathbb{I}(\Theta, R)= R(I+K)R^{-1}$. It is observed that
\begin{equation}\label{Jrel}
   A(\Theta, R,\Omega_r, \Omega)= \mathbb{I}^{-1} (\Theta, R) J(\Theta, R,\Omega_r, \Omega)
\end{equation}
where $J: TQ \to \mathfrak{g}$ is defined in \eqref{Jdef}. Details of this result can be found in \cite{marskrish} and \cite{mars}. Therefore, from \eqref{Jrel}, the momentum map in rigid body frame is
\begin{equation}\label{Jmap}
  \Pi = \mathbb{I} ((I+K)^{-1}K \Omega_r + \Omega)= (I+K) \Omega+K \Omega_r
\end{equation}
and the momentum map in the inertial frame gives the conserved angular momentum which is
\begin{equation}\label{Jmapin}
  J(\Theta, R,\Omega_r, \Omega) = \mathbb{I}(\Theta, R) A(\Theta, R,\Omega_r, \Omega)= R \Pi
\end{equation}
Here we assume the reference trajectory is generated by another rigid body with $3$ rotors.
If the spatial angular momentum of the system is $\mu$, any reference trajectory must lie in the $\mu$ level set of $J$. In other words, the set of reference trajectories is given as
\begin{equation}\label{reftraj}
  S=\{(R_d, \Theta_d, \Omega_d, \Omega_{rd})\in TQ :J(\Omega_d, \Omega_{rd}) = \mu \}
\end{equation}
where $\mu$ of the system is given by \eqref{Jmap} and $\Omega_d= T_ {R_d} L_{R_d^{-1}} \dot{R}_d$.

\section{AGAT control}
In this section we first state the result from \cite{anrnb2} for AGAT of a fully actuated rigid body with external actuators and subsequently extend it to AGAT for the rigid body with $3$ internal rotors.
\begin{definition}\label{navfn}
A function $\psi: G \to  \mathbb{R} $ on a Lie group $G$ is a navigation function (\cite{kodi}) if
\begin{enumerate}
  \item $\psi$ has a unique minimum.
  \item All critical points of $\psi$ are non-degenerate; $Det(Hess \psi(q)) \neq 0$ whenever $\mathrm{d}\psi(q) = 0$ for $q \in G$.
\end{enumerate}
\end{definition}
\begin{definition}
 The configuration error on a Lie group $G$ is the map $E: G \times G \to G$ defined as
\begin{equation} \label{errmap}
E(g, g_r) = L_{g_r}g^{-1}.
\end{equation}
\end{definition}
\begin{definition}
  Consider a Lie group $G$ and a navigation function $\psi:G \to \mathbb{R}$. The configuration error map $E$ is compatible with a navigation function $\psi$ for the tracking problem if
  \begin{itemize}
    \item $\psi \circ E :G \times G \to \mathbb{R}$ is symmetric; or, $\psi(E(g,h))= \psi(E(h,g))$ for all $g,h \in G$.
    \item $E(e) = q_0$, where $q_0$ is the minimum of the navigation function and $e$ is the identity of $G$.
  \end{itemize}
\end{definition}
The AGAT problem for a rigid body with rotors is solved in two parts. In the first part a PID control law is proposed for AGAT of a rigid body with external actuation and in the second part a feedback control law is chosen so that the equations for a rigid body with rotors reduce to an externally actuated rigid body. The first problem is well addressed in literature (\cite{pidmtp}, \cite{dayawansa}, \cite{anrnb2}). In \cite{anrnb2}, a proportional derivative (PD) and feed-forward (FF) control law achieves AGAT of a reference trajectory on $SO(3)$. We introduce an additional integral control term to the PD+FF tracking control along the lines of \cite{pidmtp} in the following theorem.

\begin{theorem}\label{thm1}
  (AGAT for an SMS on a Lie group) Let $G$ be a compact Lie group and $\mathbb{I}: \mathfrak{g} \to \mathfrak{g}^*$ be an isomorphism on the Lie algebra. Consider the SMS on the Riemannian manifold $(G, {\mathbb{I}})$ given by \eqref{dynliegrp} and a smooth reference trajectory with bounded velocity $g_r :\mathbb{R} \to G$ on the Lie group. Let $\psi : G \to \mathbb{R}$ be a navigation function compatible with the error map in \eqref{errmap}. Then there exists an open dense set $S$ in $G \times \mathfrak{g}$ such that AGAT of $g_r$ is achieved for all $(g(0), \xi(0)) \in S$ with $u = \mathbb{I}^\sharp(F)$ in \eqref{dynliegrp} given by the following equation
  \begin{align}\label{thm1eqn}
    u &= - g^{-1}_r \mathbb{G}_{\mathbb{I}}^\sharp (-k_p \mathrm{d}\psi (E) - k_d \dot{E} - k_I \xi_I) g + g^{-1} (\stackrel{\mathfrak{g}}{\nabla}_\eta \eta  \\ \nonumber
    &+  \frac{\mathrm{d}}{\mathrm{d}t}{E^{-1}}\mathrm{d}_2E(\dot{g}_r)) g - I^\sharp ad^*_\xi I \xi
  \end{align}
  where $\eta:= T_E L_{E^{-1}} \dot{E}$, $k_p$, $K_d$ and $k_I$ are constants to be chosen as shown in Appendix A and $\xi_I$ is defined as
  \begin{equation}\label{xiI}
    \stackrel{\mathbb{G}_{\mathbb{I}}}{\nabla}_{\dot{E}} \xi_I = \mathbb{G}_{\mathbb{I}}^\sharp \mathrm{d} \psi (E).
  \end{equation}
\end{theorem}
\begin{proof}
  Appendix  ~\ref{appA}.
\end{proof}

\begin{theorem}\label{thm2}
  (AGAT for a rigid body with $3$ rotors) Consider the rigid body with $3$ rotors in \eqref{redynrigrot} and a smooth, bounded reference trajectory $R_d :\mathbb{R} \to SO(3)$ so that $(R_d, \Omega_d) \in S$ given by \eqref{reftraj} for some $\Omega_{rd} \in \mathbb{R}^3$. Let $\psi : G \to \mathbb{R}$ be a navigation function compatible with the error map in \eqref{errmap}. Then there exists an open dense set $P$ in $G \times \mathfrak{g}$ such that AGAT of $g_r$ is achieved for all $(g(0), \xi(0)) \in P$ with $u_{int}$ in \eqref{redynrigrot} given by the following equation
  \begin{align}\label{thm2eqn}
    u_{int} &= -\breve{u_{ext}} +K(\Omega+\Omega_r)\times \Omega
  \end{align}
  where
  \begin{align}\label{thm2eqn2}
  u_{ext} &=-R^{-1}({I}^{-1} (-k_p skew(PE) -k_d E^{-1} \dot{E}-k_I E^{-1} \xi_I)\\ \nonumber
  &- {I}^{-1} ad^*_\eta {I}\eta)R +\widehat{\dot{\Omega_d}}+ \widehat{[\Omega, \Omega_d]} - {I}^\sharp ad^{*}_{\widehat{\Omega}} {I}\widehat{\Omega},
  \end{align} $P$ is a positive definite symmetric matrix, $E= R_d R^{-1}$ and $\xi_I$ is defined in \eqref{xiI}.
\end{theorem}

\begin{proof}
  Substituting for $K \dot{\Omega}_r$ from the second equation in \eqref{dynrigrot} to the first
  \begin{equation}\label{step1}
  (I+K) \dot{\Omega} + u_{int} -K \dot{\Omega} = \Pi \times \Omega
  \end{equation}
  Therefore,
  \begin{equation}\label{step2}
    I \dot{\Omega}-\Pi \times \Omega  = -u_{int}
  \end{equation}
  From \eqref{Jmap},
  \begin{equation}\label{step3}
     I \dot{\Omega}- I \Omega \times \Omega = -u_{int} +K(\Omega+ \Omega_r) \times \Omega
  \end{equation}
  From \eqref{dynliegrp}, the left hand side is a SMS on $SO(3)$ similar to \eqref{dynrigext} where $\Omega$ is the body velocity of the rigid body and the right hand side is the control field which is to be designed for AGAT of the rigid body. Therefore, theorem ~\ref{thm1} is applicable. As the objective is AGAT of the rigid body and not the rotors, we set $\breve{u_{ext}} = -u_{int} +K(\Omega+ \Omega_r) \times \Omega$ where $u_{ext} \in \mathfrak{so(3)}$ is obtained from \eqref{thm1eqn} by choosing $\psi(R) =trace(P(I-R))$ for a positive definite symmetric matrix $P$ and $E \coloneq R_d R^{-1}$ is the configuration error defined in \eqref{errmap}. It can be shown that $\psi$ is a navigation function compatible with $E$. Details of this result can be found in \cite{anrnb2}. The rotor dynamics is given by substituting for $\dot{\Omega}$ from \eqref{step2} in \eqref{redynrigrot} as
  \begin{equation}\label{step4}
    \dot{\Omega}_r = (I+K)^{-1}u_{int} - I^{-1}\Pi \times \Omega
  \end{equation}
\end{proof}
\textit{Remark}: In \cite{pidmtp}, the externally actuated rigid body is allowed to have bounded parametric uncertainty in inertia and actuation models and AGAT is achieved for the proposed PID control law. For the rigid body with rotors, however, the presence of bounded parametric uncertainty and bounded constant disturbances leads to semi-global convergence as shown in \cite{madhumtp} for interconnected mechanical systems.

\section{Simulation results}
Consider a rigid body with rotors having the following parameters $I  = \begin{pmatrix} 4 & 1 & 1\\ 1 & 5.2 & 2 \\ 1 & 2 & 6.3 \end{pmatrix}$,$ K = diag(5,6,7)$ and the following initial conditions $R(0) = \begin{pmatrix} 0.36&  0.48 &-0.8 \\ -0.8 &0.6& 0\\ 0.48 & 0.64 & 0.60\\ \end{pmatrix}$, $\Omega(0)= I^{-1} \begin{pmatrix} 1 & 2.2 & 5.1 \end{pmatrix}$ and $\Omega_r(0)= \begin{pmatrix}.5 & 1.9 & 1.5 \end{pmatrix}$. The reference trajectory is generated by a dummy rigid body with rotors having the following parameters $ I_d= \begin{pmatrix}
               1 & 0& 0\\ 0 &1.2& 0\\ 0& 0 &2 \\
               \end{pmatrix}$, $K_d = diag(4,3,2)$ and the following initial conditions $R_d(0)= id(3)$, $\Omega_{d}(0) = I_d^{-1} \begin{pmatrix} -0.8 & -0.3 & -0.5 \end{pmatrix}^T$, $\Omega_{rd}(0)$ given by the momentum conservation equation \eqref{reftraj} and, a constant input vector field $u_{int} = \begin{pmatrix}
               0;0;0
               \end{pmatrix}^T$ is applied. The dynamics of the dummy rigid body is given by \eqref{redynrigrot}. In order to find $K_p$, $K_d$ and $K_I$ we use the bounds in Appendix ~\ref{appA}. $\mu = 2 *(\lambda_{min}(I)+\lambda_{max}(I))/\lambda_{min}(I)= 2.0048$, $\lambda = 2 \frac{\lambda_{max}(I)} {\lambda_{min}(I)^2}= 1.42$. We choose $K_I=1$, $K_d =3$ and $K_p =1$ to find $u_1$ in \eqref{thm1eqn} and subsequently $u$ in \eqref{thm2eqn}. The simulation results are shown in figure ~\ref{fig13}. (https://www.dropbox.com/s/7ji88vri7mp2pxd/circle.avi?dl=0 ) is a video link showing $3$-D tracking of the axis representation in the quaternion representation. In both the cases, the reference trajectory is in red and the controlled trajectory is in blue and the plots show both the trajectories in matrix representation of $SO(3)$.
\newline
Now the reference trajectory is generated by the same dummy rigid body with $u_{int} = \begin{pmatrix}
       sin(t) & cos(t) & sin(t)
     \end{pmatrix}$ and  $u_{int} = \begin{pmatrix}
       0.2 & 0.1 & 0.2
     \end{pmatrix}$. The simulations are shown in figures ~\ref{fig7} and ~\ref{fig12} respectively.
\newline
We compare the control effort by considering the $2-$ norm of $\breve{u_{ext}}$ in \eqref{thm2eqn2} with the AGAT control for a rigid body in \cite{pidmtp}. The trace function is considered as a navigation function on $SO(3)$ with the same $P$, $k_d$, $k_p$ and $k_I$ values for both the simulations. The trajectories for tracking the same reference are plotted in figure ~\ref{fig15}. The control law for internal actuation is obtained from \eqref{thm2eqn} for the AGAT tracking law in \cite{pidmtp} and compared with the proposed control in figure ~\ref{fig17}.
  \begin{figure}[ht]
\centering
\begin{subfigure}[b]{0.22\textwidth}
  \includegraphics[scale=0.15]  {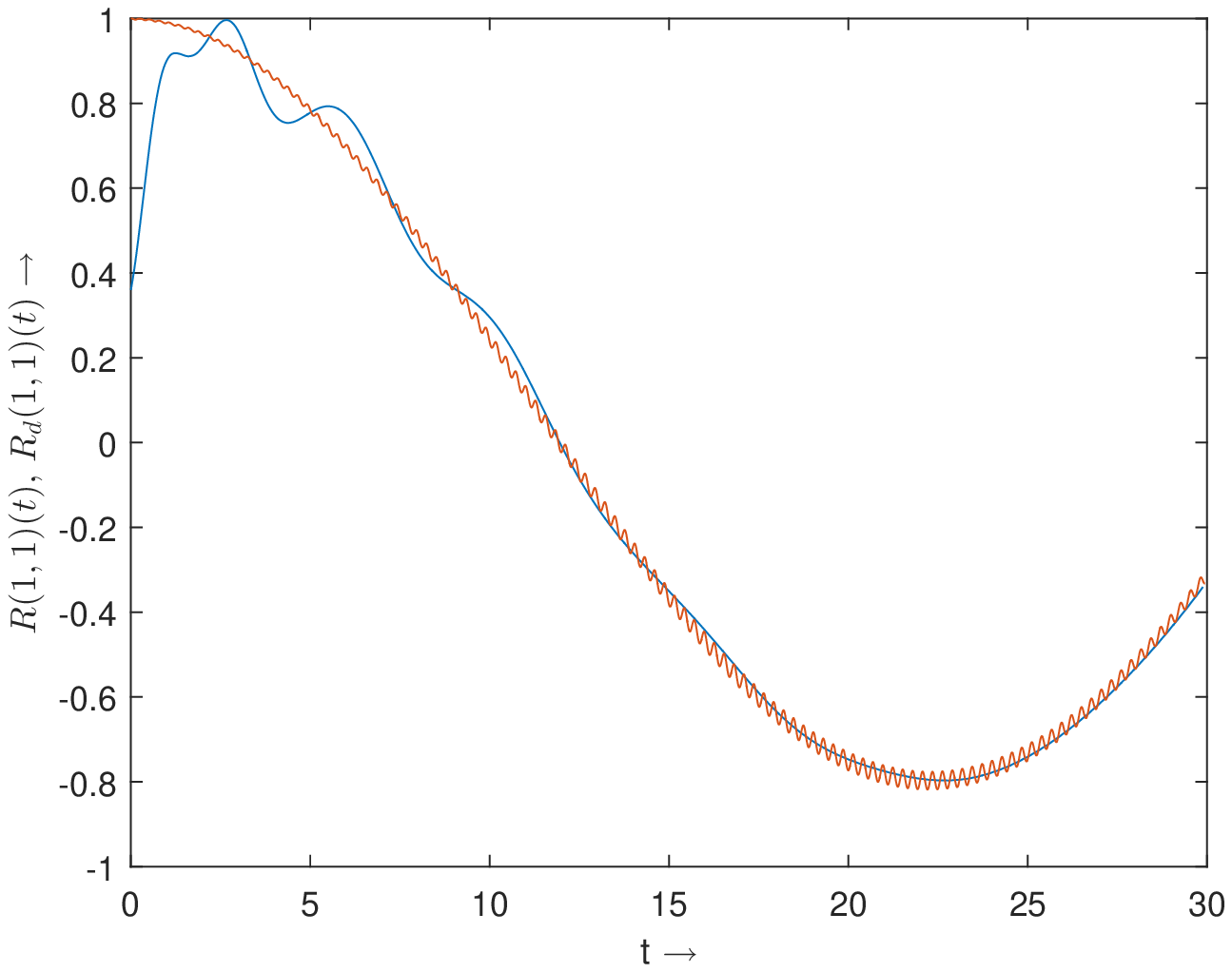}

  \caption{$(1,1)$(t) of both trajectories}
  \end{subfigure}
\begin{subfigure}[b]{0.22\textwidth}
  \includegraphics[scale=0.15]{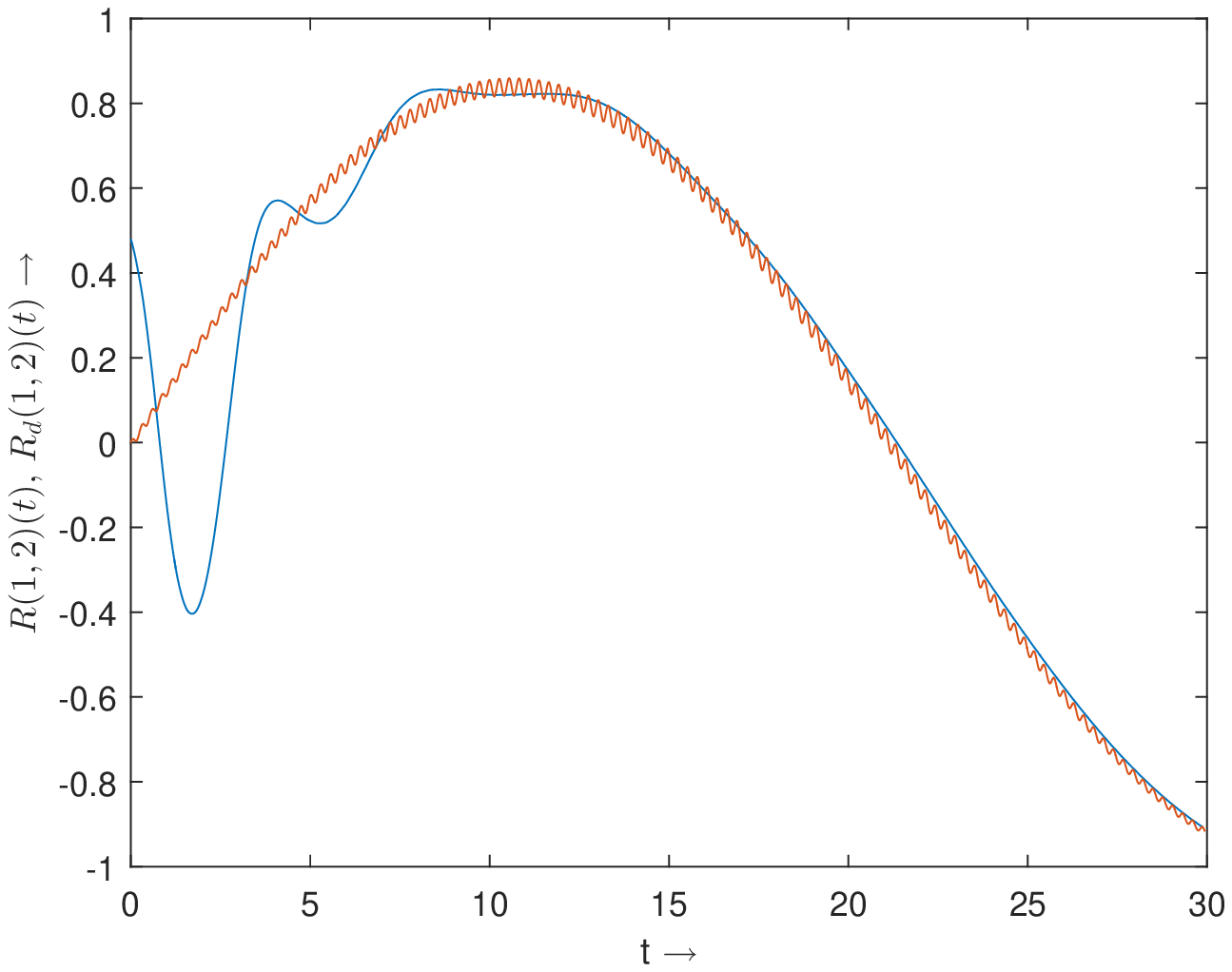}
  \caption{$(1,2)$(t) of both trajectories}
\end{subfigure}
\begin{subfigure}[b]{0.22\textwidth}
  \includegraphics[scale=0.15]{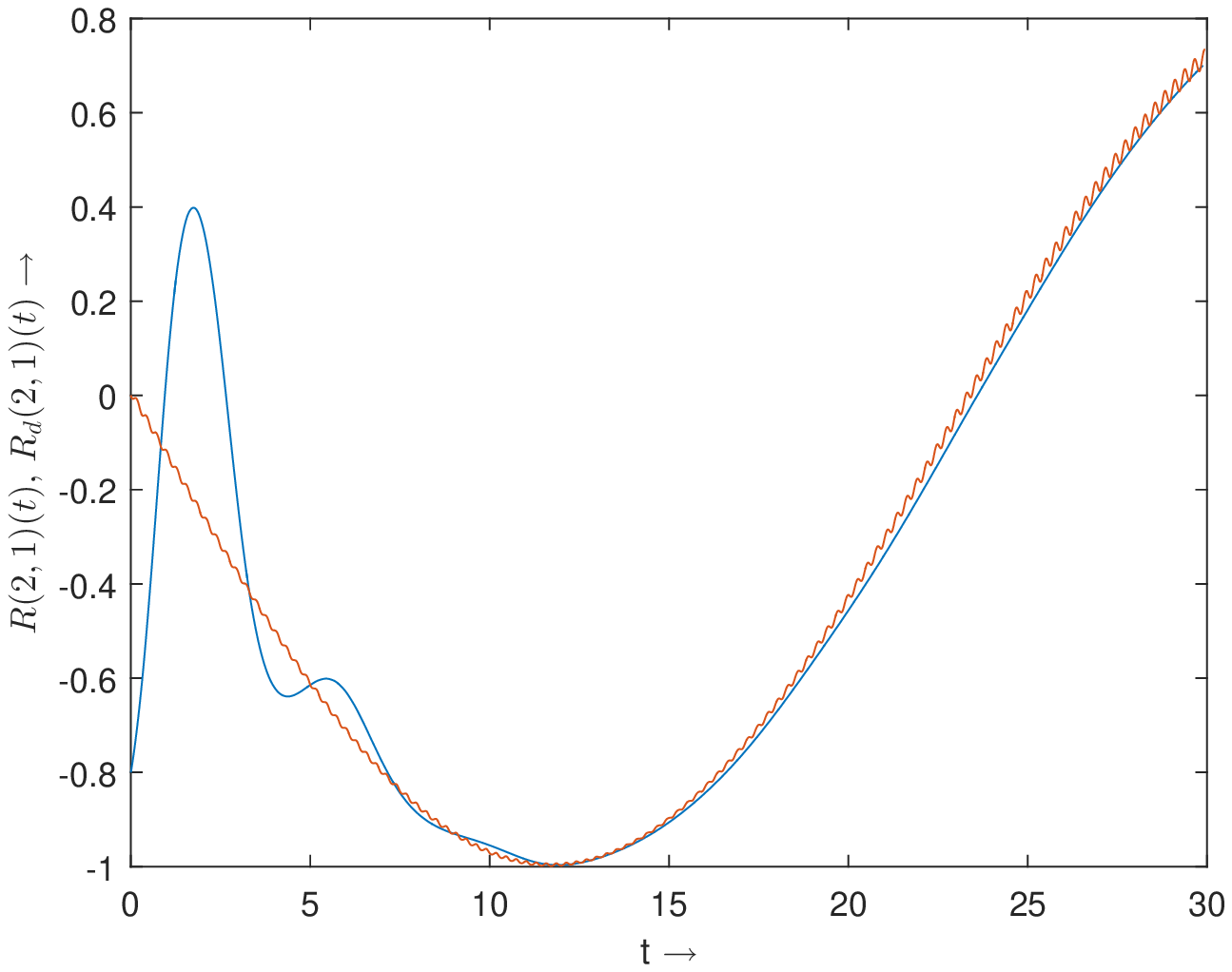}
  \caption{$(2,1)$(t) of both trajectories}
\end{subfigure}
\begin{subfigure}[b]{0.22\textwidth}
  \includegraphics[scale=0.15]{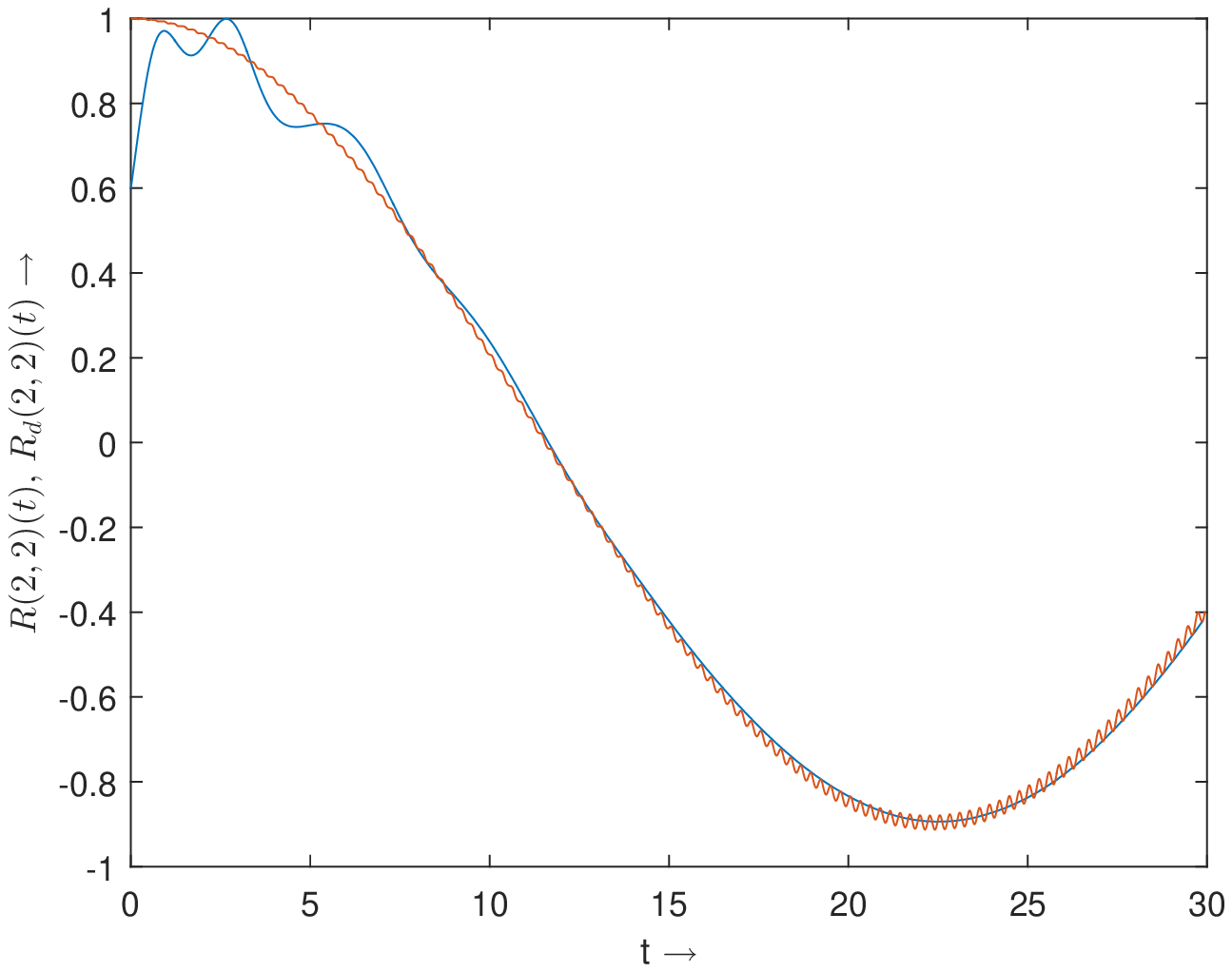}
  \caption{$(2,2)$(t) of both trajectories}
\end{subfigure}
  \caption{Representative plots for $u_{int} = \protect \begin{pmatrix}
                 0 ; 0 ; 0
               \protect \end{pmatrix}^T$ as input torque to reference rigid body} \label{fig13}
\end{figure}

\begin{figure}[ht]
\centering
\begin{subfigure}[b]{0.22\textwidth}
  \includegraphics[scale=0.15]  {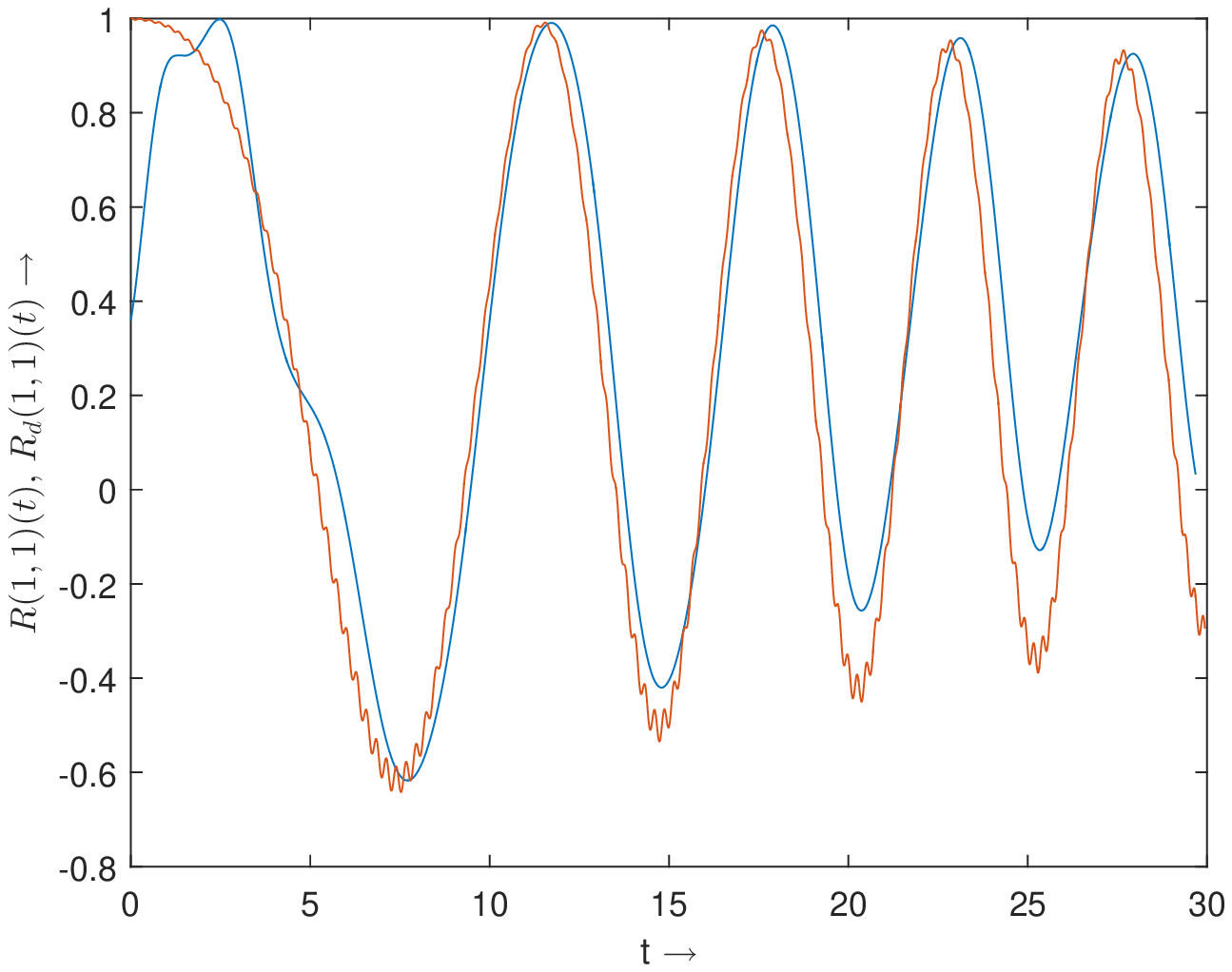}

  \caption{$(1,1)$(t) }
  \end{subfigure}
\begin{subfigure}[b]{0.22\textwidth}
  \includegraphics[scale=0.15]{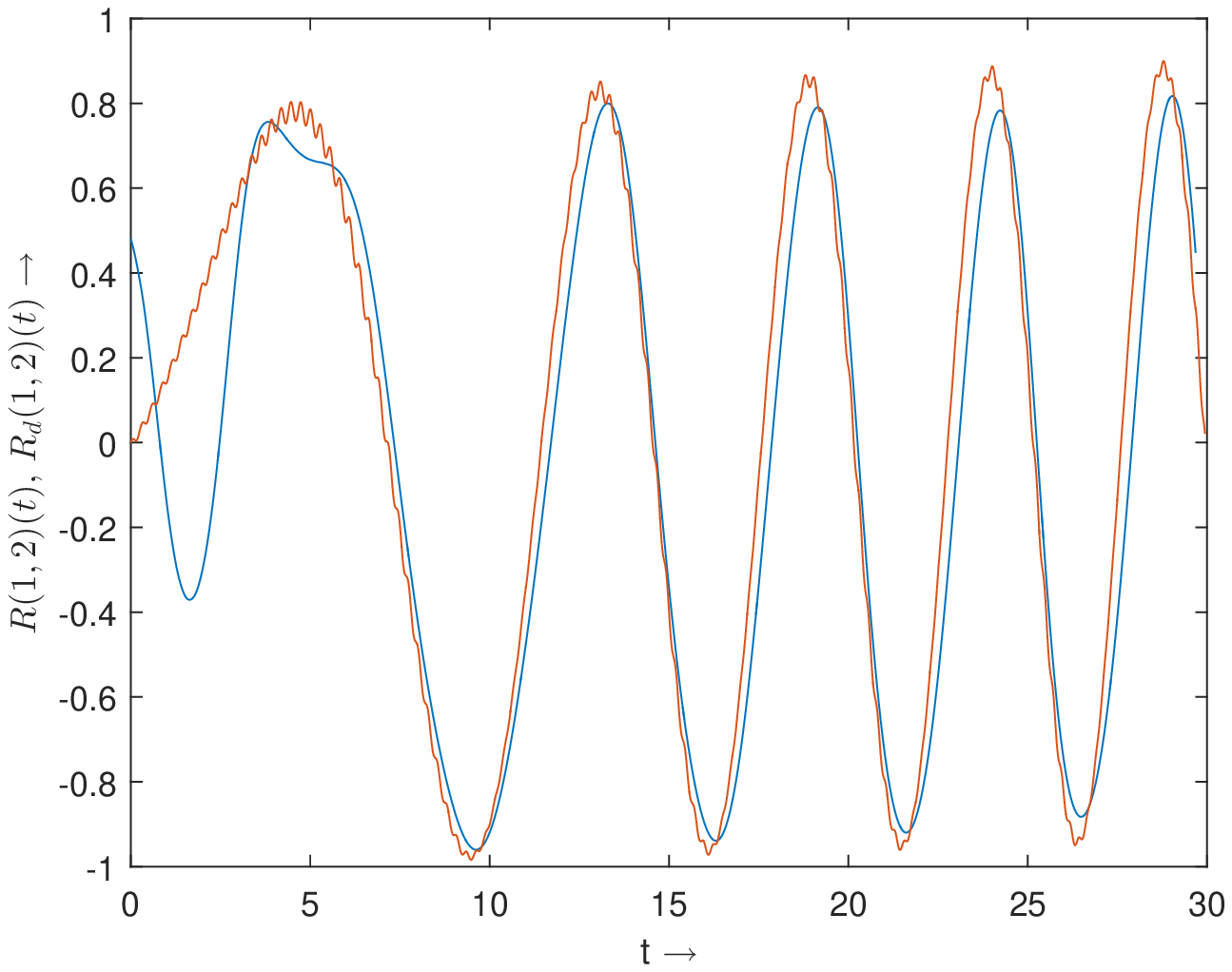}
  \caption{$(1,2)$(t) }
\end{subfigure}
\begin{subfigure}[b]{0.22\textwidth}
  \includegraphics[scale=0.15]{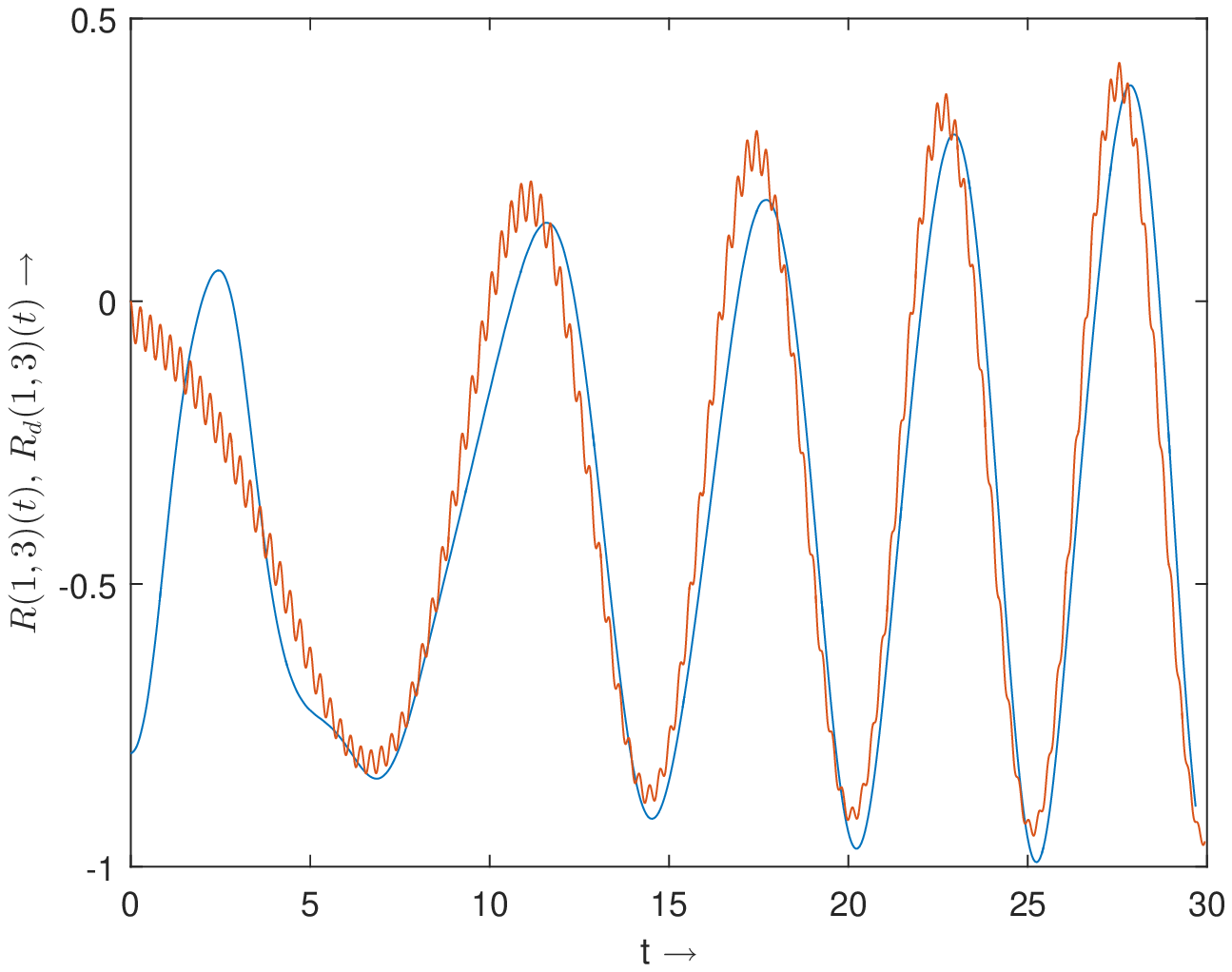}
  \caption{$(1,3)$(t)}
\end{subfigure}
\begin{subfigure}[b]{0.22\textwidth}
  \includegraphics[scale=0.15]  {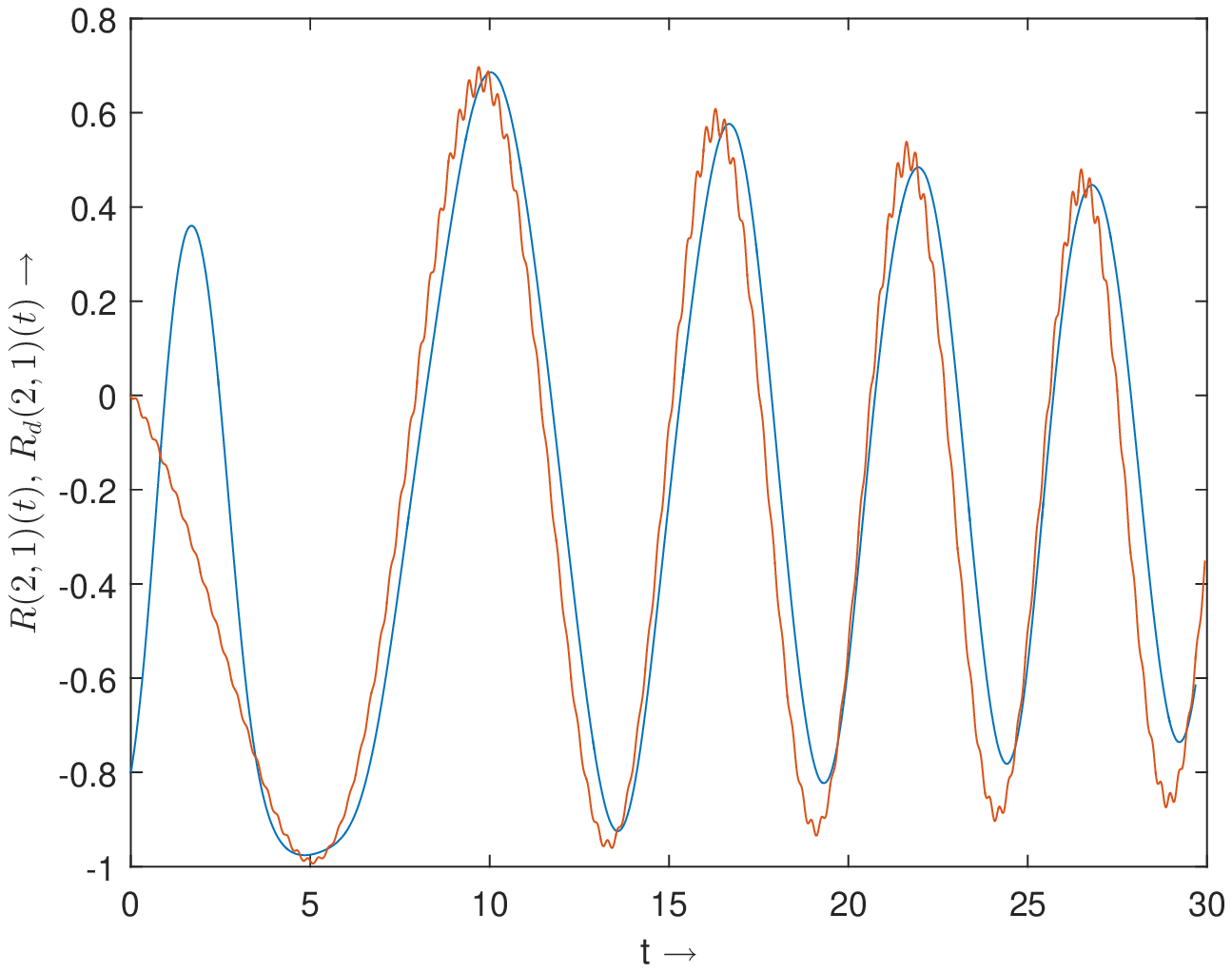}
  \caption{$(2,1)$(t) }\label{fig10}
  \end{subfigure}
 \caption{Representative plots for $u_{int} = \protect \begin{pmatrix}
                 0.2 ; 0.1 ; 0.2
               \protect \end{pmatrix}^T$ as input torque to reference rigid body} \label{fig7}
\end{figure}

\begin{figure}[ht]
\centering
\begin{subfigure}[b]{0.22\textwidth}
  \includegraphics[scale=0.15]  {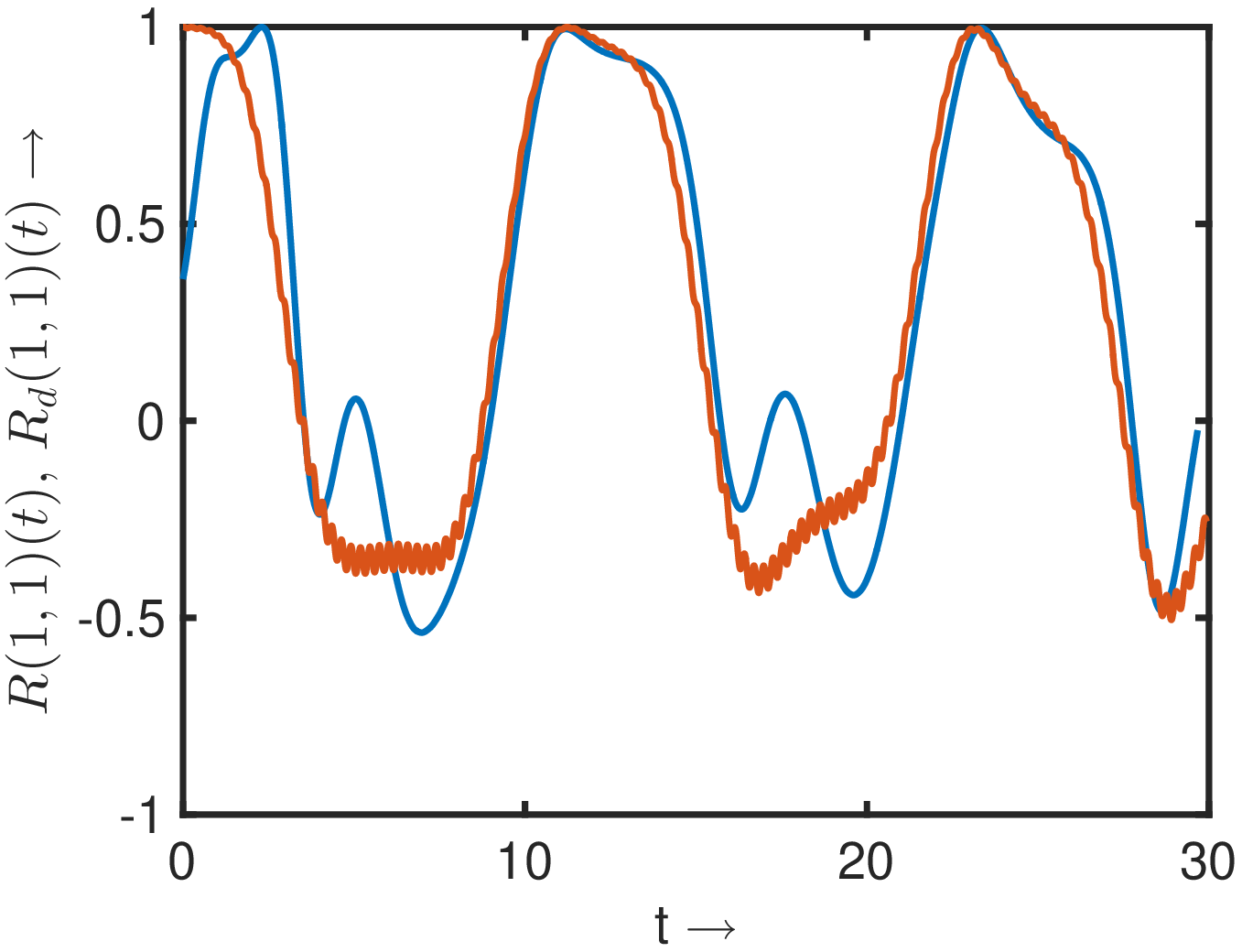}

  \caption{$(1,1)$(t)}
  \end{subfigure}
\begin{subfigure}[b]{0.22\textwidth}
  \includegraphics[scale=0.15]{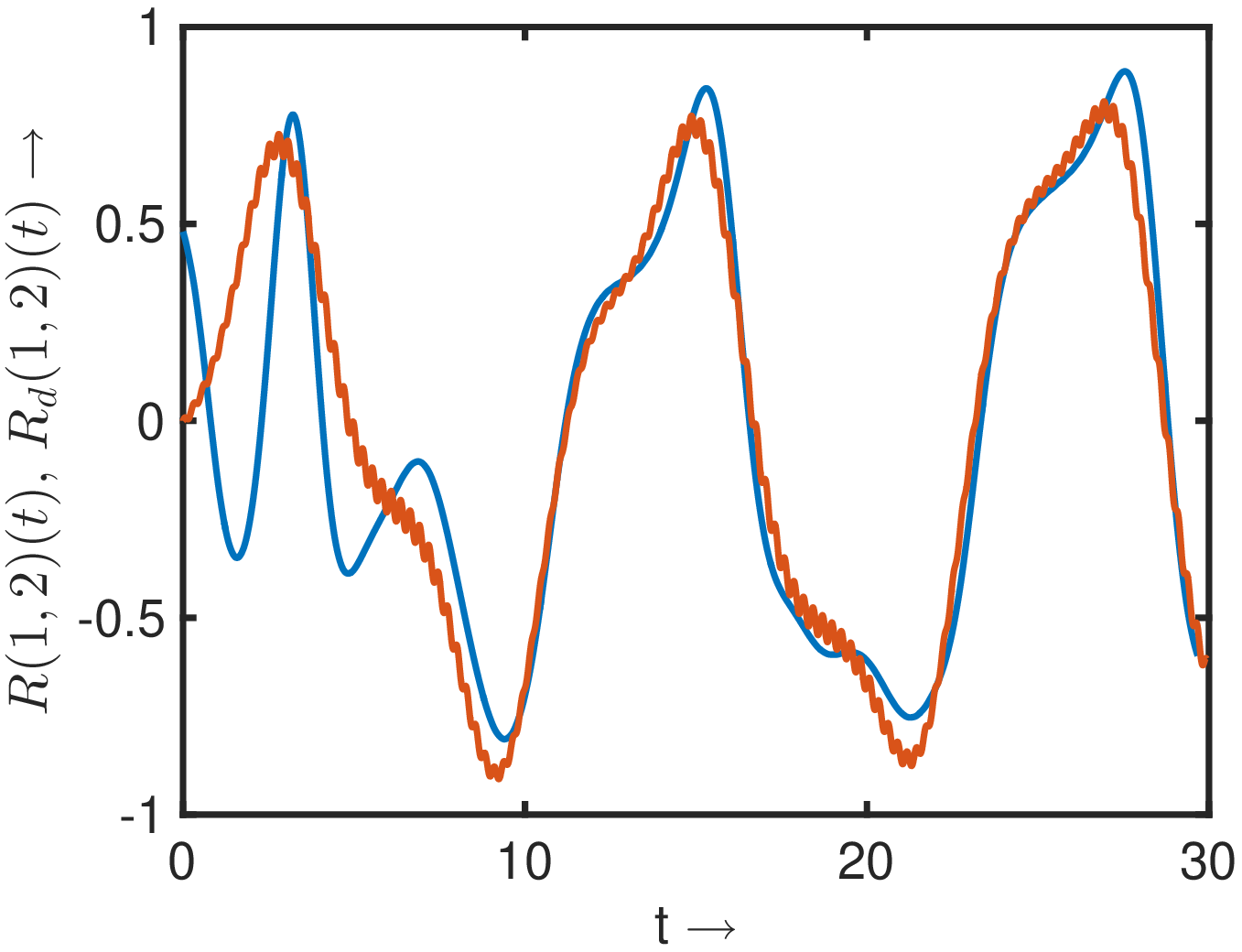}
  \caption{$(1,2)$(t)}
\end{subfigure}
\begin{subfigure}[b]{0.22\textwidth}
  \includegraphics[scale=0.15]{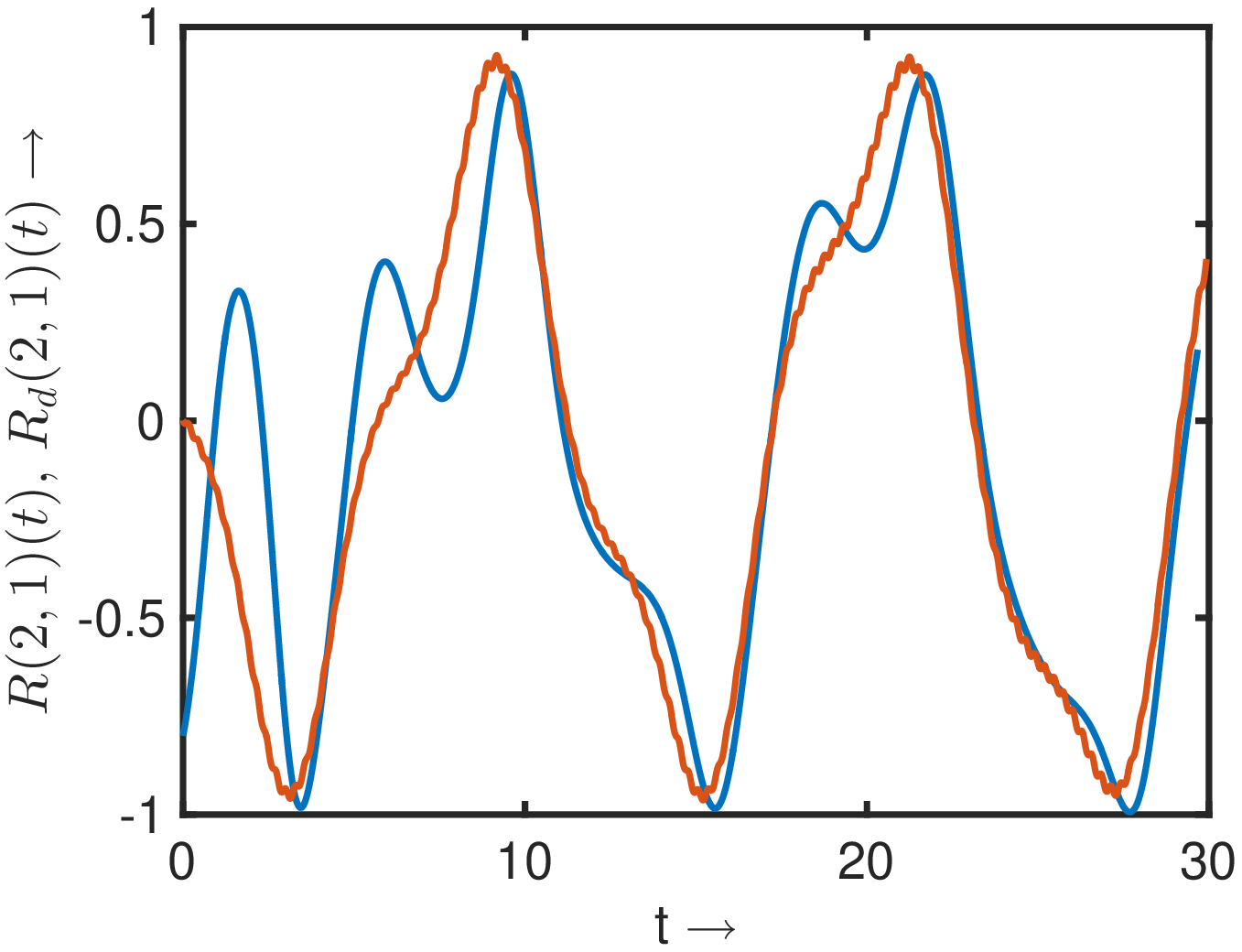}
  \caption{$(2,1)$(t)}
\end{subfigure}
\begin{subfigure}[b]{0.22\textwidth}
  \includegraphics[scale=0.15]{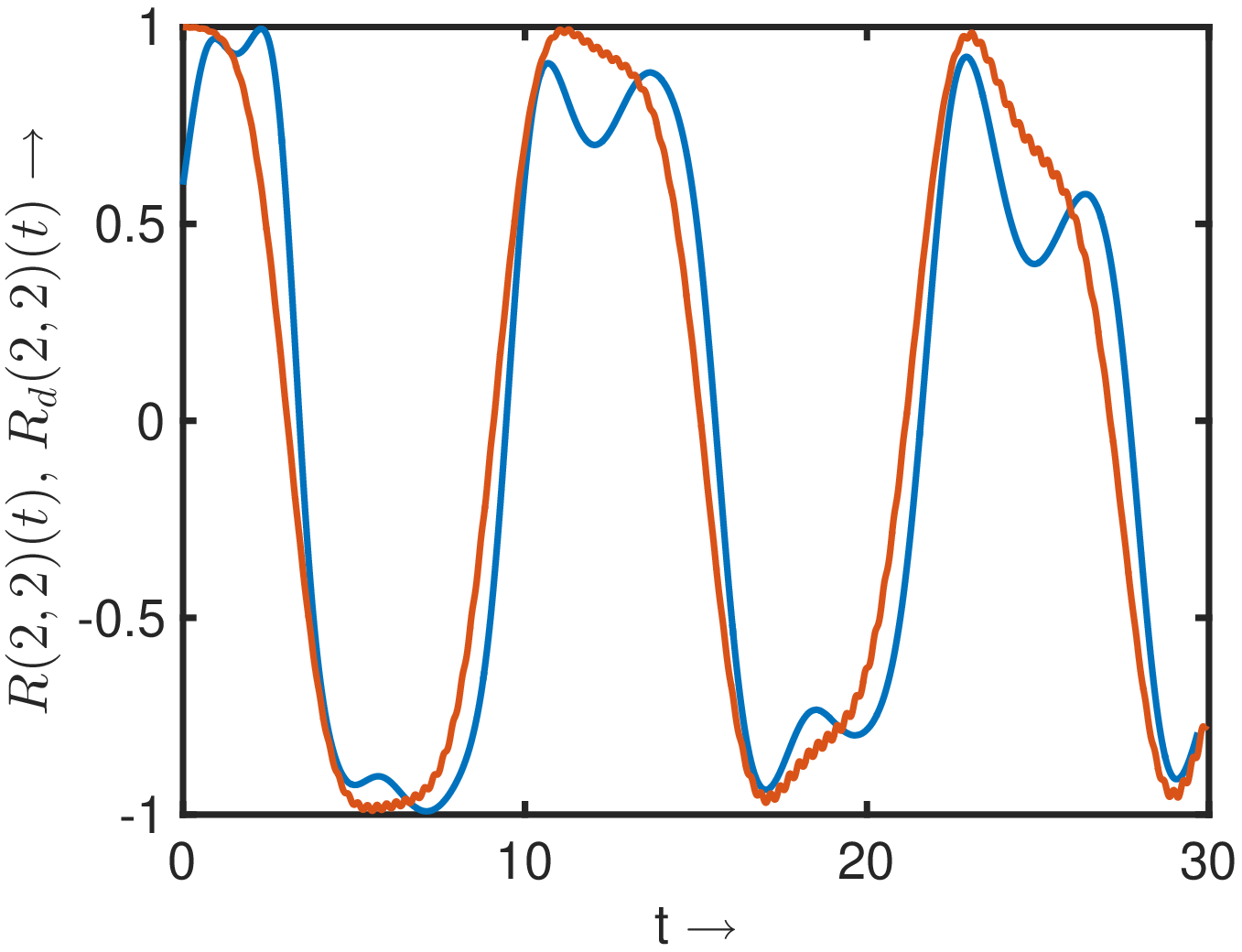}
  \caption{$(2,2)$(t) }
\end{subfigure}
  \caption{Representative plots for $u_{int} = \protect \begin{pmatrix}
                 \sin{t} ; \cos{t} ; \sin{t}
               \protect \end{pmatrix}^T$ as input torque to reference rigid body} \label{fig12}
\end{figure}

\begin{figure}[ht]
 \centering
  \begin{subfigure}[b]{0.15\textwidth}
  \includegraphics[scale=0.075]  {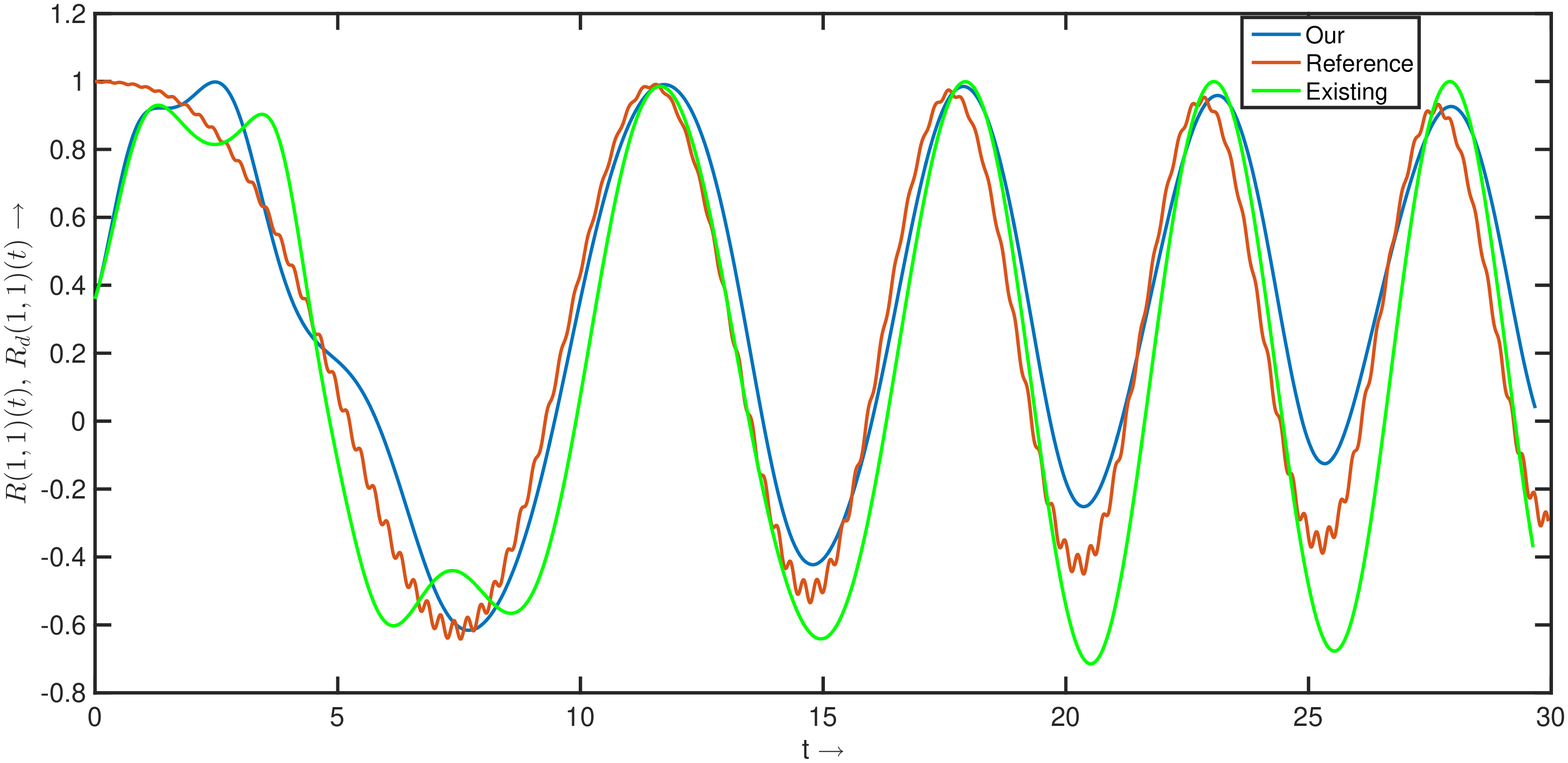}
  \end{subfigure}
\begin{subfigure}[b]{0.15\textwidth}
  \includegraphics[scale=0.075]{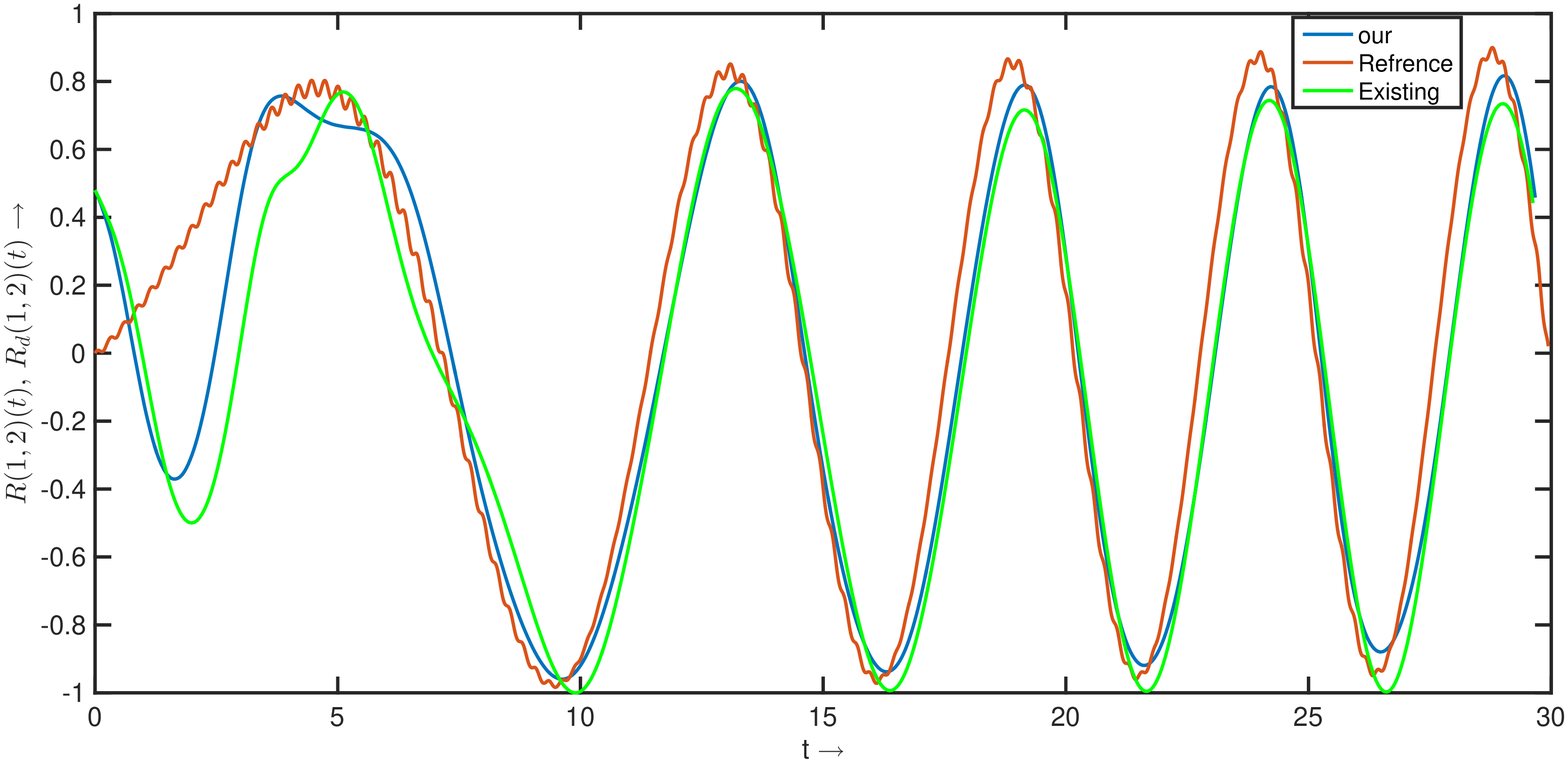}
\end{subfigure}
\begin{subfigure}[b]{0.15\textwidth}
  \includegraphics[scale=0.075]{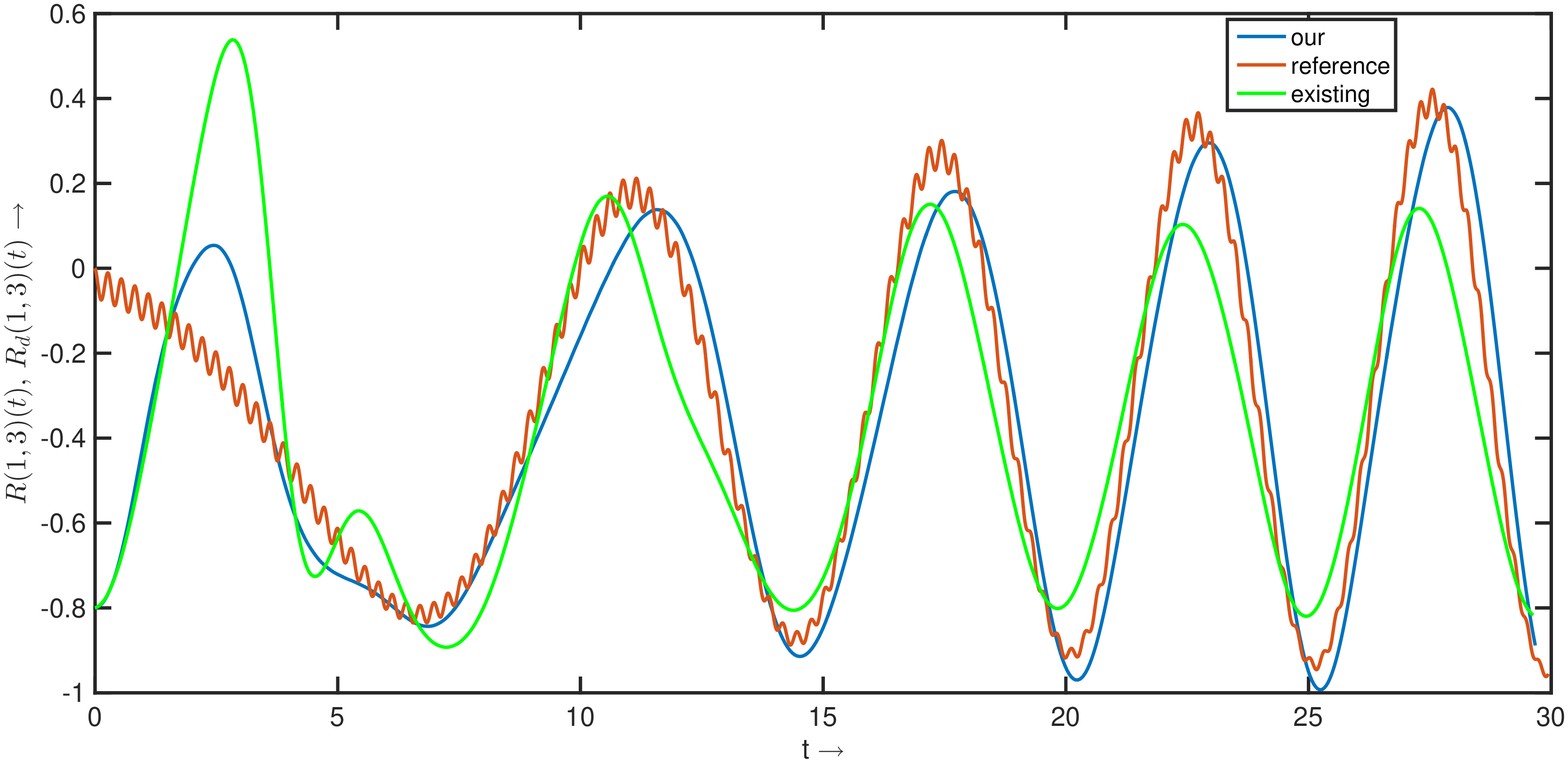}
\end{subfigure}
 \begin{subfigure}[b]{0.15\textwidth}
  \includegraphics[scale=0.075]  {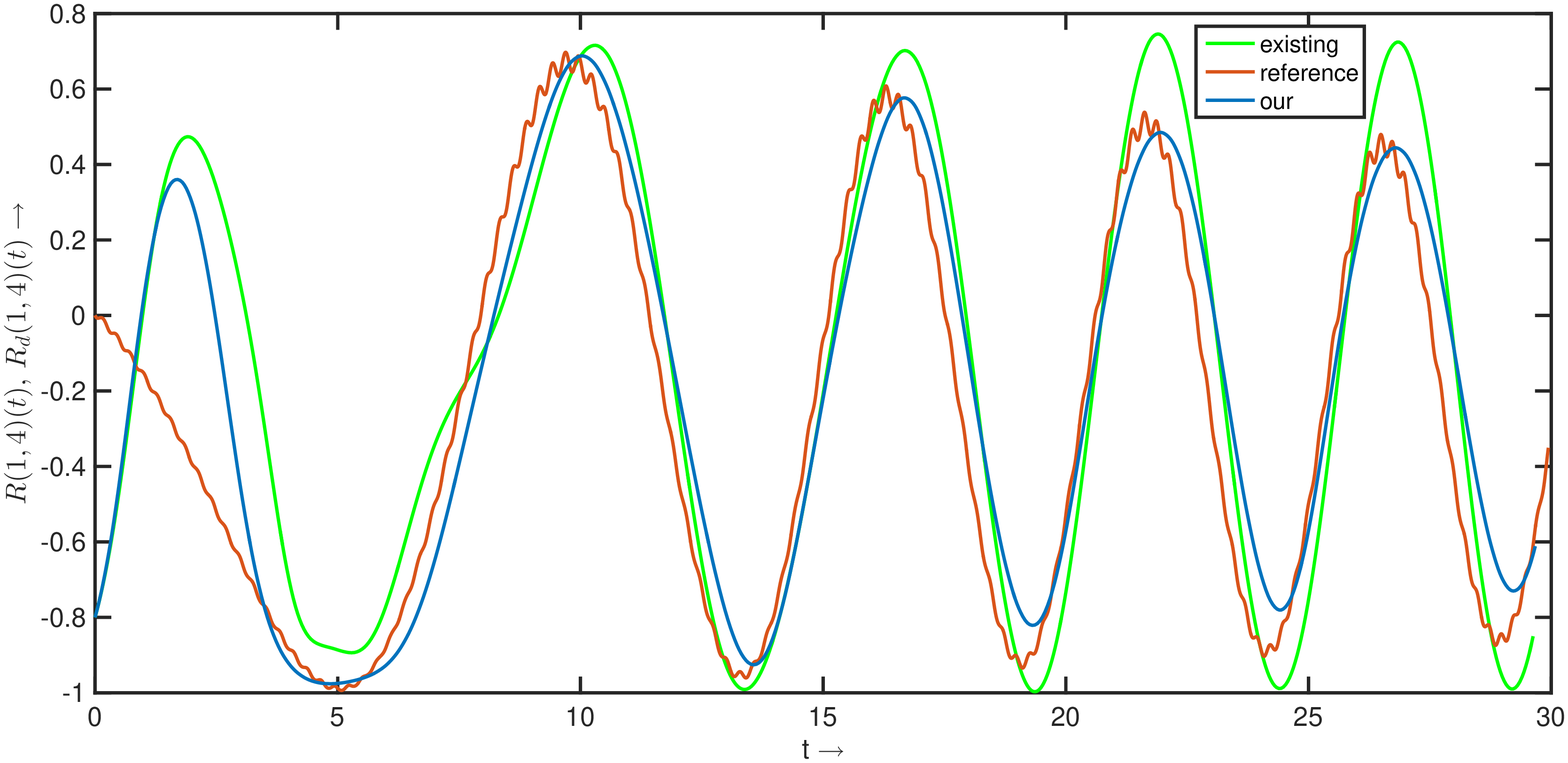}
  \end{subfigure}
\begin{subfigure}[b]{0.15\textwidth}
  \includegraphics[scale=0.075]{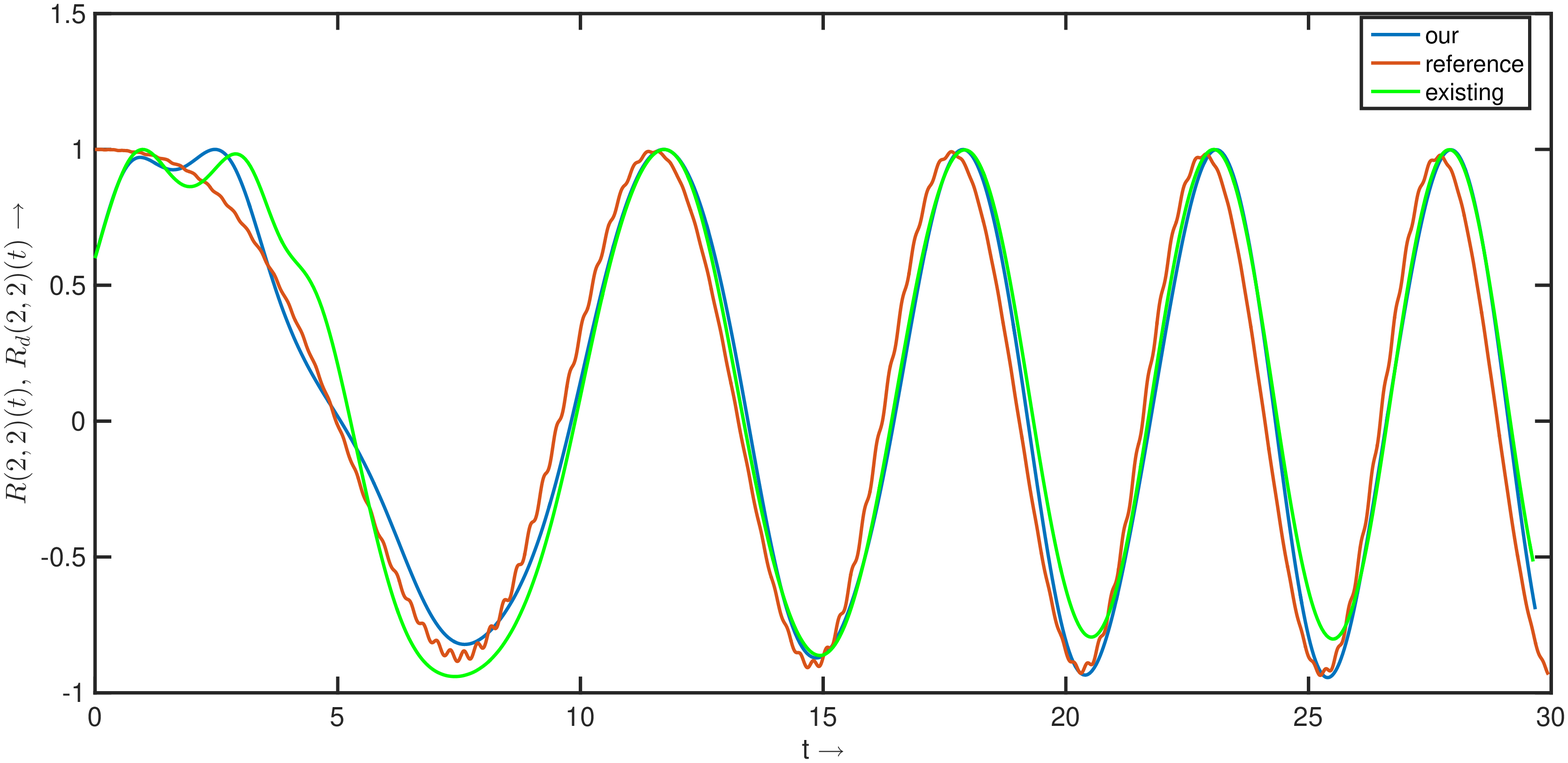}
\end{subfigure}
\begin{subfigure}[b]{0.15\textwidth}
  \includegraphics[scale=0.075]{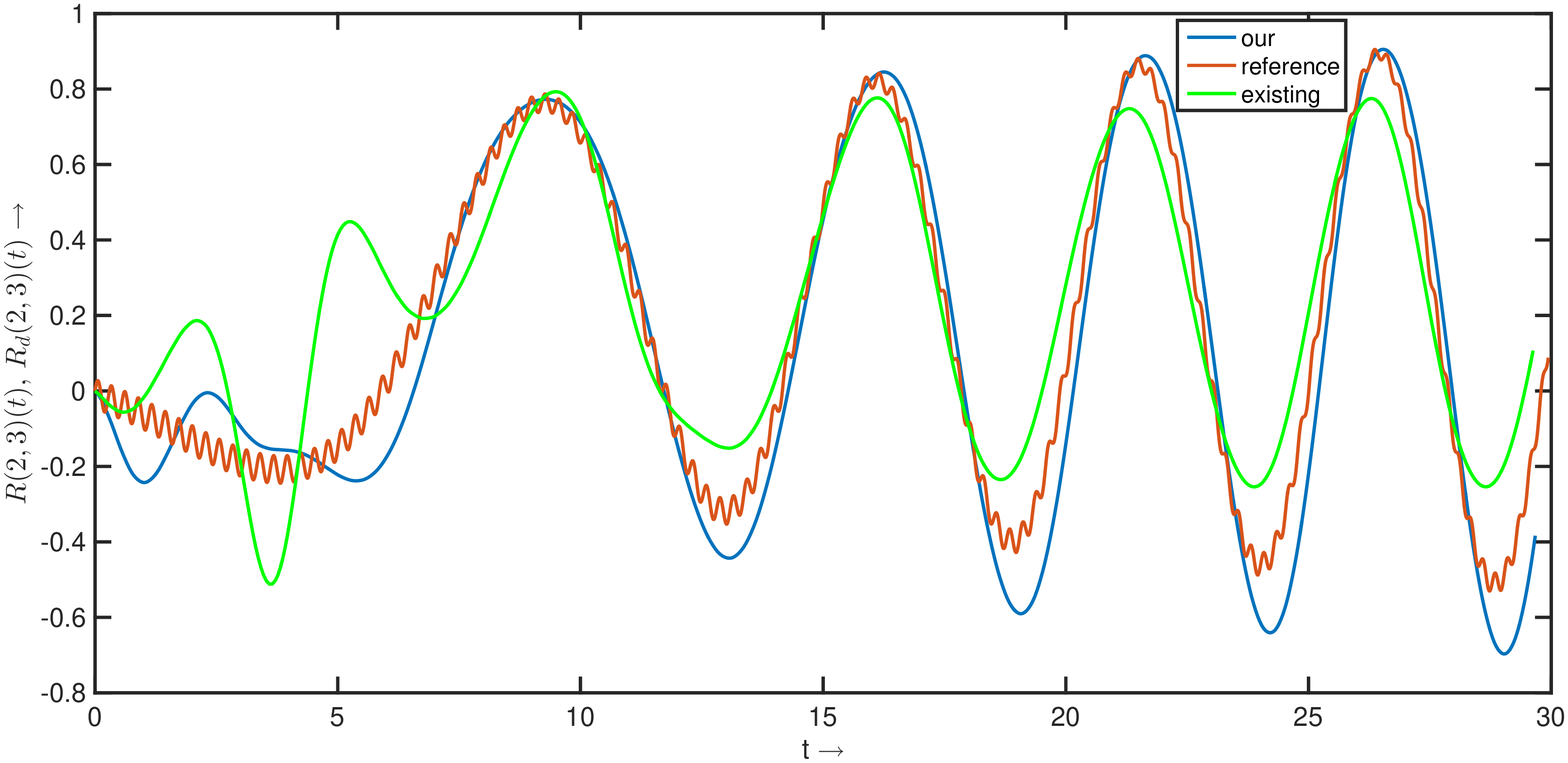}
\end{subfigure}
  \caption{Comparison of tracking results for $u_{int} = \protect \begin{pmatrix}
                 0.2 ; 0.1 ; 0.2
               \protect \end{pmatrix}^T$ as input torque to reference rigid body} \label{fig15}
  \end{figure}

 \begin{figure}[ht]
  \centering
   \begin{subfigure}[b]{0.22\textwidth}
  \includegraphics[scale=0.2]{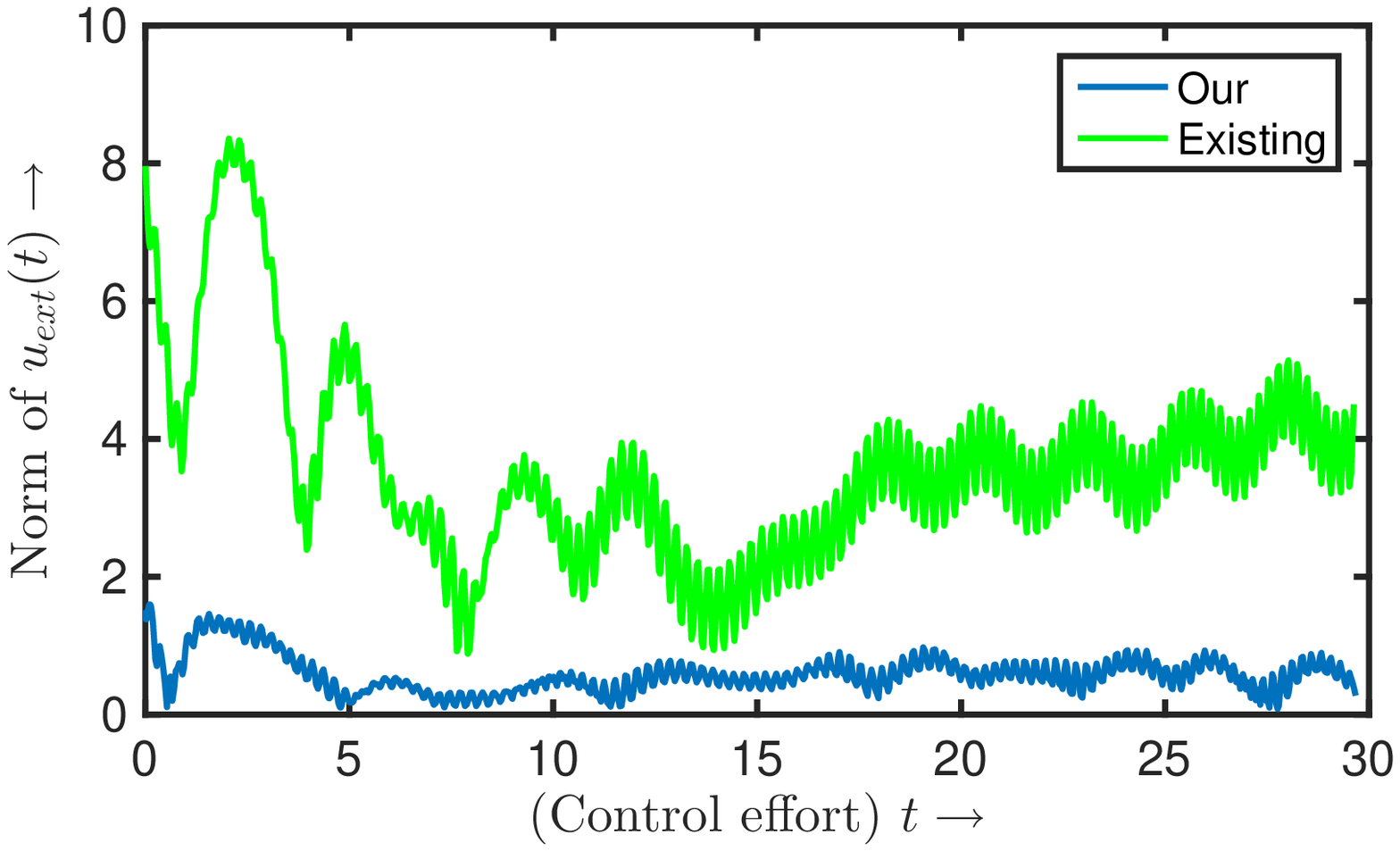}
  \caption{ $u_{ext}$}
\end{subfigure}
\begin{subfigure}[b]{0.22\textwidth}
  \includegraphics[scale=0.2]{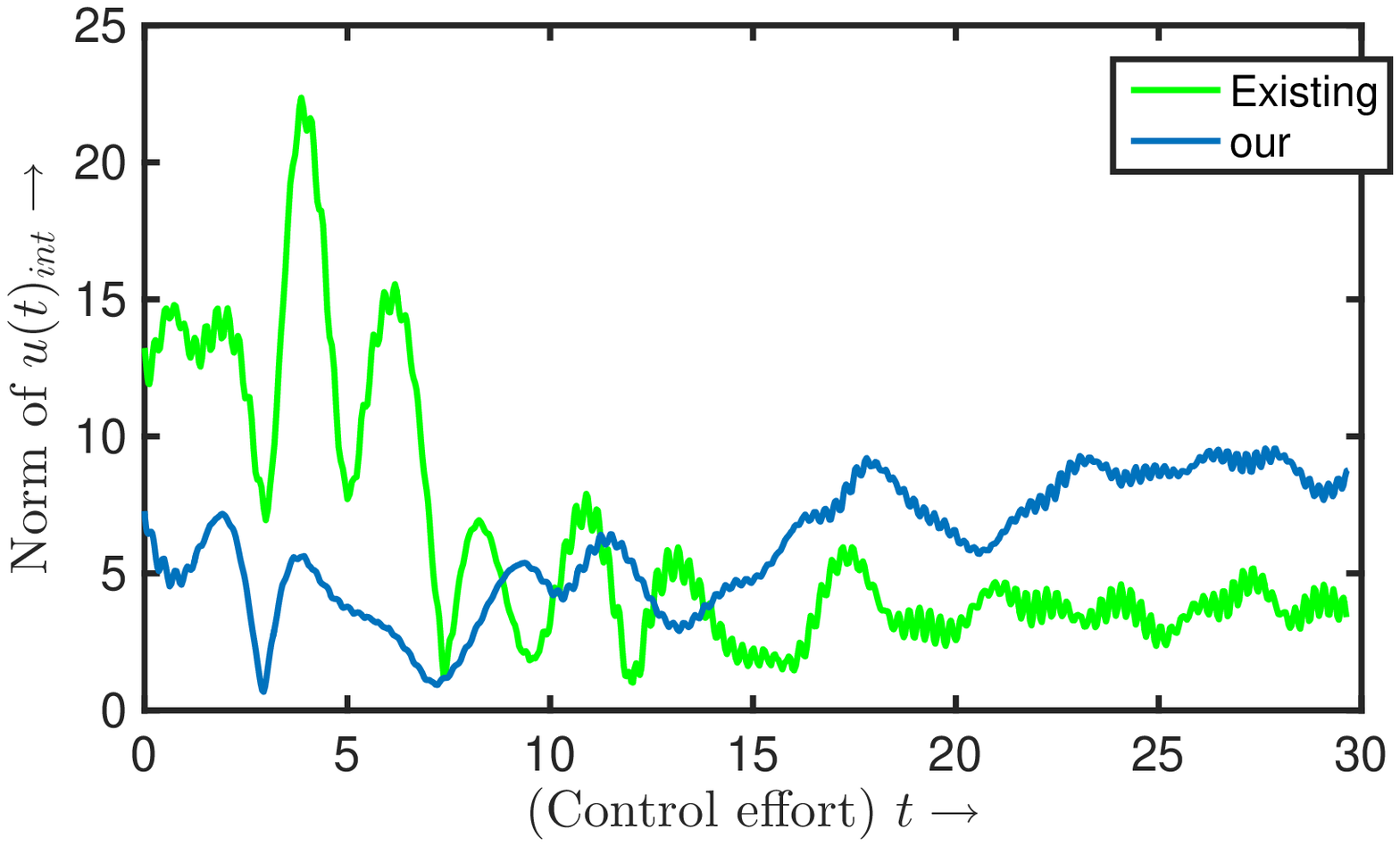}
  \caption{ $u_{int}$}
\end{subfigure}
\caption{Comparison of $2-$ norm of control effort for $u_{int} = \protect \begin{pmatrix}
                 0.2 ; 0.1 ; 0.2
               \protect \end{pmatrix}^T$ as input torque to reference rigid body}\label{fig17}
\end{figure}


\appendices
\section{}\label{appA}
We define $v_e \coloneq \dot{E}$ where $E(g(t), g_r(t))$ is the error trajectory on $G$ defined in \eqref{errmap}. The error dynamics for the SMS in \eqref{dynliegrp} is
\begin{equation}\label{errdyn}
  {\stackrel{\mathbb{G}_{\mathbb{I}}}{\nabla}}_{\dot{E}(g(t), g_{r}(t))}{v_e} =  -k_p\mathbb{G}^\sharp \mathrm{d}\psi (E) - k_d \dot{E} -k_I {\xi_I}
\end{equation}
 where $\xi_I$ is defined in \eqref{xiI}. Therefore, $u_i \in T_E G$. Consider $E_{cl}: TG \times TG \to \mathbb{R}$ given as
\begin{align}\label{Ecl}
E_{cl}(E, \dot{E} ,\mathbb{G}^\sharp \mathrm{d}\psi (E))&= \frac{1}{2} \ll\dot{E},\dot{E} \gg + k_p \psi (E) \\ \nonumber
&+ \frac{\tau}{2}\ll \xi_I, \xi_I\gg \\ \nonumber
&+ \alpha\ll \mathbb{G}^\sharp \mathrm{d}\psi (E),\dot{E} \gg
+ \beta \ll \xi_I, \dot{E} \gg \\ \nonumber
&+ \delta \ll\mathbb{G}^\sharp \mathrm{d}\psi (E), \xi_I \gg
\end{align}
for some constants $k_p$, $k_I$, $k_d$, $\alpha$, $\beta$, $\tau$, $\delta$ that will be determined shortly. The following shows $E_{cl}$ is negative semi-definite. We essentially follow the proof outlined in \cite{pidmtp}.\\
\textit{Remark}: The objective is to show that $E_{cl}$ is non-increasing along trajectories to \eqref{redynrigrot}. The derivative of the first term is given by \eqref{errdyn}. Therefore, appropriately weighted cross terms in \eqref{Ecl} must be introduced so that $E_{cl}$ is negative semi-definite.
\begin{align*}
\frac{\mathrm{d}}{\mathrm{d}t} E_{cl} &= \ll \nabla_{\dot{E}} \dot{E}, \dot{E} \gg +k_p \ll \mathbb{G}^\sharp \mathrm{d}\psi (E), \dot{E} \gg \\ \nonumber
&+ \tau \ll \nabla_{\dot{E}} \xi_I, \xi_I \gg  + \alpha \ll \nabla_{\dot{E}}\mathbb{G}^\sharp \mathrm{d}\psi (E), \dot{E}  \gg \\ \nonumber
&+ \alpha \ll \mathbb{G}^\sharp \mathrm{d}\psi (E), \nabla_{\dot{E}}\dot{E}\gg + \beta \ll \nabla_{\dot{E}} \xi_I, \dot{E}\gg \\ \nonumber
&+ \beta \ll \xi_I, \nabla_{\dot{E}}\dot{E} \gg + \delta \ll \nabla_{\dot{E}}\mathbb{G}^\sharp \mathrm{d}\psi (E), \xi_I\gg
\\ \nonumber
&+ \delta \ll \mathbb{G}^\sharp \mathrm{d}\psi (E), \nabla_{\dot{E}}\xi_I\gg\\ \nonumber
\end{align*}
Grouping together terms and using the fact that $ \nabla_{\dot{E}} \xi_I = \mathbb{G}^\sharp \mathrm{d}\psi (E)$, we have,
\begin{align*}
\frac{\mathrm{d}}{\mathrm{d}t} E_{cl} &= \ll \nabla_{\dot{E}} \dot{E}, \dot{E} + \alpha \mathbb{G}^\sharp \mathrm{d}\psi (E) + \beta \xi_I \gg \\ \nonumber
&+ k_p \ll \mathbb{G}^\sharp \mathrm{d}\psi (E), \dot{E} \gg \\ \nonumber
&+ \ll \nabla_{\dot{E}}\mathbb{G}^\sharp \mathrm{d}\psi (E), \alpha \dot{E} +\delta \xi_I \gg \\ \nonumber
&+  \ll\mathbb{G}^\sharp \mathrm{d}\psi (E), \tau \xi_I + \beta \dot{E} + \delta \mathbb{G}^\sharp \mathrm{d}\psi (E) \gg \\ \nonumber
\end{align*}
From \eqref{errdyn},
\begin{align*}
\frac{\mathrm{d}}{\mathrm{d}t} E_{cl} &= -k_d \ll \dot{E}, \dot{E} \gg +(\beta-k_d\alpha) \ll \dot{E},\mathbb{G}^\sharp \mathrm{d}\psi (E)\gg \\ &-(k_d\beta +k_I )\ll \dot{E}, \xi_I \gg \\ \nonumber
&+ (\delta-\alpha k_p )\ll \mathbb{G}^\sharp \mathrm{d}\psi (E), \mathbb{G}^\sharp \mathrm{d}\psi (E) \gg \\
&+(\tau -k_p \beta +\alpha k_I) \ll \mathbb{G}^\sharp \mathrm{d}\psi (E), \xi_I \gg \\ \nonumber
&- \beta k_I \ll \xi_I , \xi_I \gg  +\ll \nabla_{\dot{E}}\mathbb{G}^\sharp \mathrm{d}\psi (E), \alpha \dot{E} +\delta \xi_I \gg\\ \nonumber
\end{align*}
The $Hess (\psi(q))$ is a $(0,2)$ tensor on $T_q G$ defined as $Hess(\psi (q)) (X,Y) =\ll  \nabla_{X}\mathbb{G}^\sharp \mathrm{d}\psi (q), Y \gg$ for $X$,$Y \in T_q G$, $q \in G$.
As $\psi$ is a navigation function with a unique minimum at $q_0 \in G$, there exists an compact neighborhood $S_{q_0}$ of $(q_0,0,0)$ in $TG \times TG$ in which the $Hess(\psi)$ is postive-definite and $||Hess(\psi)||_2$ is bounded. This implies there is a $\mu >0 $ such that $\ll \nabla _{X}\mathbb{G}^\sharp \mathrm{d}\psi(q),Y \gg < \mu \ll X,Y \gg$ for all $(q,X,Y) \in S_{q_0}$. Therefore,
\begin{align*}
\ll \nabla_{\dot{E}}\mathbb{G}^\sharp \mathrm{d}\psi (E), \alpha \dot{E} +\delta \xi_I \gg \quad &< \quad  \mu \alpha \ll \dot{E}, \dot{E}\gg \\
&+ \mu \delta \ll \dot{E}, \xi_I \gg
\end{align*}
So,
\[ \frac{\mathrm{d}}{\mathrm{d}t} E_{cl} \leq- v Q v^T
\]
where $v = \begin{pmatrix}
             || \dot{E}|| & || \mathbb{G}^\sharp \mathrm{d}\psi (E)|| & ||\xi_I||
           \end{pmatrix}$
and \begin{align*}Q &= \\
 &\begin{pmatrix}
          k_d  -\mu \alpha & \frac{1}{2}( k_d \alpha-\beta) & \frac{1}{2}(k_d \beta + k_I-\mu \delta) \\
         \frac{1}{2}( k_d \alpha-\beta)& \alpha k_p-\delta & \frac{1}{2}( k_p \beta +\alpha k_I-\tau ) \\
          \frac{1}{2}(k_d \beta + k_I-\mu \delta)  & \frac{1}{2}( k_p \beta +\alpha k_I-\tau ) & \beta k_I
        \end{pmatrix}\end{align*}
We shall now express $\alpha$, $\beta$, $\delta$ and $\tau$ in terms of $k_p$, $k_d$ and $k_I$. We set $\alpha = \frac{\beta}{k_d}$ so that $Q_{12}= Q_{21}=0$. Next, we choose $\tau = k_p \beta +\alpha k_I$ makes $Q_{23} = Q_{32}=0$. Let $\beta = \frac{k_I}{k_d} $ and $\delta = 2\kappa k_I$ such that $ \frac{1}{\mu} < \kappa < \frac{2}{\mu}$. This makes $Q_{13}= Q_{31}= \mu \delta - 2k_I = -\sigma k_I$ where $\sigma > 0$. \\
$Q_{11} = k_d- \mu \alpha  = k_d-\mu \frac{k_I}{k^2_d}$,\\
$Q_{22} = \frac{k_I }{k^2_d}(k_p-2 \kappa k_d^2)$ and
$Q_{33} = \frac{k_I^2}{k_d}$. Then,  \begin{equation*}
                                            Q =\begin{pmatrix}
          k_d-\mu \frac{k_I}{k^2_d} &0 & -\sigma k_I\\
         0& \frac{k_I }{k^2_d}(k_p-2 \kappa k_d^2)& 0 \\
          -\sigma k_I  & 0 & \frac{k_I^2}{k_d}
        \end{pmatrix}.
                                          \end{equation*}
It is observed that the leading principal minors of $Q$ are positive if $k_p > 2 \kappa k_d^2$ and $0 < k_I< \frac{k_d^3}{\mu}(1-\sigma^2)$. Let  $ \lambda \coloneq \mathrm{sup}_{S_{q_0}} \frac{\ll \mathbb{G}^\sharp \mathrm{d}\psi(E), \mathbb{G}^\sharp \mathrm{d}\psi(E) \gg} {2 \psi(E) }$. It can be shown (\cite{pidmtp}) that $E_{cl}$ is positive definite for all $(E, \dot{E},\mathbb{G}^\sharp \mathrm{d}\psi(E))  \in S_{q_0}$ with $k_p > \lambda\frac{(\sigma^2 + \tau \alpha^2)}{(\tau - \beta^2)} = \lambda\frac{4 \kappa^2 k_I k_d^6 + k_I^3 + k_I^2 k_d k_p}{k_p k_d^5}$. Therefore, we choose $k_p > max \Bigg  \{ 2 \kappa k_d^2, \frac{\lambda k_I^2}{2 k_d^4} \left (1+ \sqrt{1+ \frac{4 k_d^3(k_I^2 + 4 \kappa^2 k_d^6)}{\lambda k_I^3}} \right ) \Bigg \}$ so that both $Q$ and $E_{cl}$ are positive definite.
\newline
The error dynamics \eqref{errdyn} is a dissipative SMS with a control vector field proportional to gradient of a navigation function. From the result in \cite{kodi}, therefore, $u_i$ achieves AGAS of the error dynamics. From \cite{anrnb2}, AGAT control for \eqref{dynrigext} is given by
\begin{align*} u_{ext}&= - g^{-1}_r \mathbb{G}_{\mathbb{I}}^\sharp (u_i) g + g^{-1} (\stackrel{\mathfrak{g}}{\nabla}_\eta \eta  +  \frac{\mathrm{d}}{\mathrm{d}t}{E^{-1}}\mathrm{d}_2E(\dot{g}_r)) g \\ \nonumber
    &- I^\sharp ad^*_\xi I \xi
\end{align*}

\section*{ACKNOWLEDGMENT}
A part of this work was carried out when the second author was visiting IIT-Gandhinagar. Ravi N Banavar acknowledges with pleasure the support provided by IIT-Gandhinagar.

\bibliographystyle{plain}
\bibliography{aps1}
\end{document}